\documentclass{article}

\title{Marketplace Operators Can Induce Competitive Pricing}  

\author{Tiffany Ding\thanks{University of California, Berkeley, tiffany\_ding@berkeley.edu. Work done while interning at Amazon.}, Dominique Perrault-Joncas\thanks{Amazon, joncas@amazon.com}, Orit Ronen\thanks{Amazon, oritron@amazon.com}, \\ Michael I. Jordan\thanks{University of California, Berkeley and Inria, Paris, jordan@cs.berkeley.edu}, Dirk Bergemann\thanks{Yale University and Amazon, dirk.bergemann@yale.edu}, Dean Foster\thanks{Amazon, dean@foster.net}, Omer Gottesman\thanks{Amazon, omergott@amazon.com}} 
\date{ }

\usepackage[utf8]{inputenc} 
\usepackage[T1]{fontenc}    
\usepackage{booktabs}       
\usepackage{nicefrac}         
\usepackage{microtype}      
\usepackage[dvipsnames]{xcolor}         
\usepackage[round]{natbib}

\usepackage{amssymb,amsmath,amsthm,bbm}
\usepackage{verbatim,float,url,dsfont}
\usepackage{algorithm,algorithmic}
\usepackage{mathtools}
\usepackage{multirow,multicol}
\usepackage{ragged2e}
\usepackage{xr-hyper}
\usepackage{array}
\usepackage{stmaryrd,scalerel}

\ifdefined\NeurIPS  
\usepackage{hyperref}       
\else 
\usepackage[colorlinks=true,citecolor=blue,urlcolor=blue,linkcolor=blue]{hyperref}
\usepackage[margin=1in]{geometry}
\fi

\usepackage{bbm}
\usepackage{graphicx}
\usepackage{amsmath,amssymb,amsthm,amsfonts}
\usepackage{url}
\usepackage{makecell} 
\usepackage{tikz}
\usetikzlibrary{arrows.meta}

\usepackage[shortlabels]{enumitem} 
\usepackage{subcaption} 
\usepackage{etoc} 
\usepackage{algorithm,algorithmic} 
\usepackage[most]{tcolorbox} 
\usepackage{pgfplots} 
\usetikzlibrary{patterns} 
\allowdisplaybreaks 
\usepackage{wrapfig} 

\newcommand{\I}[1]{\mathbbm{1}{\left\{#1\right\}}} 

\def\E{{\mathbb E}}
\def\P{{\mathbb P}}
\def\R{{\mathbb R}}

\def\cW{{\mathcal W}}

\newcommand{\pd}[1]{\frac{\partial}{\partial #1}} 

\newtheorem{theorem}{Theorem}
\newtheorem{lemma}{Lemma}

\newtheorem{proposition}{Proposition}
\theoremstyle{definition}
\newtheorem{remark}{Remark}
\newtheorem{definition}{Definition}
\newtheorem{assumption}{Assumption}

\usepackage[capitalise]{cleveref}
\Crefname{equation}{Eq.}{Eqs.}
\Crefname{figure}{Figure}{Figs.}
\Crefname{tabular}{Table}{Tabs.}
\Crefname{table}{Table}{Tabs.}
\Crefname{section}{Section}{Secs.}

\usepackage[]{color-edits}
\addauthor{td}{red}
\addauthor{dpj}{cyan}
\addauthor{og}{orange}

\newcommand{\IS}{\mathtt{I}}
\newcommand{\MO}{\mathtt{M}}
\newcommand{\qthresh}{q^{\dagger}}
\newcommand{\qddag}{q^{\ddagger}}
\newcommand{\qthresheps}{\underline{q}^{\dagger}}
\newcommand{\qddageps}{\underline{q}^{\ddagger}}
\newcommand{\psole}{p^{\star}_{\IS}}
\newcommand{\psoleA}{p^{\star}_{\MO}}

\newcommand{\CS}{\mathrm{CS}}
\newcommand{\PS}{\mathrm{PS}}
\newcommand{\BR}{\mathrm{BR}}

\newcommand{\CSsole}{\CS^{\star}}
\newcommand{\uISsole}{u_{\IS}^{\star}}

\newcommand{\pBR}{p_{\IS}^{\mathrm{BR}}}
\newcommand{\qBR}{q_{\IS}^{\mathrm{BR}}}

\newtcolorbox[use counter={BoxCounter}]{game}[2][] {
title= #2  Extensive form of game,
#1
}

\begin{document}
\maketitle


\begin{abstract}
As e-commerce marketplaces continue to grow in popularity, it has become increasingly important to understand the role and impact of marketplace operators on competition and social welfare. We model a marketplace operator as an entity that not only facilitates third-party sales but can also choose to directly participate in the market as a competing seller. 
We formalize this market structure as a price-quantity Stackelberg duopoly in which the leader is a marketplace operator and the follower is an independent seller who shares a fraction of their revenue with the marketplace operator for the privilege of selling on the platform. The objective of the marketplace operator is to maximize a weighted sum of profit and a term capturing positive customer experience, whereas the independent seller seeks solely to maximize their own profit. We derive the subgame-perfect Nash equilibrium and find that it is often optimal for the marketplace operator to induce competition by offering the product at a low price to incentivize the independent seller to match their price.
\end{abstract}

\newpage

\section{Introduction}

The rapid growth of e-commerce marketplaces has given rise to a new market structure in which the marketplace operator plays a dual role as both operator and seller: Companies such as
Amazon, Walmart, and Costco facilitate transactions for third-party sellers while selling their own products under the Amazon Basics, Great Value, and Kirkland labels.
Despite the growing prevalence of this operator-as-seller configuration, formal economic analysis is lacking.
In this work, we develop and analyze a two-player game that captures the key aspects of this type of marketplace.

How did this hybrid structure come to be, and why might we want more of it?
When marketplaces rely exclusively on third-party sellers, consumers can face inflated prices and erratic product availability. This occurs because, in the absence of direct competition, a third-party seller behaves as a monopolist and sets prices to maximize its own profit. Such behavior harms not only customers but also the marketplace operator, as the health of a marketplace depends on repeat purchases and customer trust, both of which deteriorate when prices are perceived as unfair. In this sense, the incentives of consumers and the operator are naturally aligned. When third-party sellers engage in price gouging, both sides lose. The marketplace operator therefore has a vested interest in finding mechanisms to induce sellers to set competitive prices, meaning prices below their monopolist price. One option is to impose direct regulation or price caps, but these are blunt and heavy-handed instruments. In this paper, we instead explore how marketplace operators can achieve competitive pricing through natural market forces—by strategically participating as sellers themselves. We hope that this work highlights this mechanism as a practical and underappreciated tool for fostering healthier, more competitive marketplaces.

We formulate a price-quantity Stackelberg duopoly where the first mover is the marketplace operator, who seeks to optimize not only their profit but also customer experience. The second mover is an independent seller, who sells on the marketplace and pays a commission on revenue to the marketplace operator. 
In addition to online retailers, this model can be more generally applied to physical marketplaces such as Best Buy, Target, supermarkets, and pharmacies, all of which offer private-label products that compete with name-brand versions on their shelves. Using this model, we answer the following questions: 
\begin{enumerate}
    \item \textbf{If the marketplace operator wants to induce competitive pricing, how can they do so?} We show that in order for the marketplace operator to induce meaningful competition, they must \emph{maintain enough inventory} for the independent seller to perceive them as a credible competitor. Furthermore, they should set a price lower than the independent seller's monopolist price but higher than the independent seller's break-even price (\cref{lemma:best_response_intensity}).
    \item \textbf{When is it in the marketplace operator's best interest to induce competition?} It depends on the game parameters. \emph{When the marketplace operator's cost is relatively high} compared to the independent seller's cost, the marketplace operator will prefer to induce competition. As the marketplace operator's cost decreases, it may become optimal for the operator to sell some inventory themselves and let the independent seller capture the residual demand,
    and eventually it will become optimal for the marketplace operator to sell the product themselves. Theorem \ref{theorem:equilibrium_unconstrained_game} describes the equilibrium strategies of both players and allows us to determine the exact boundaries between these regions.
    \item \textbf{How are consumers affected by the marketplace operator selling in their own market?} We show that the entry of the marketplace operator in their own market transfers surplus from the independent seller to consumers (\cref{lemma:surplus_transfer}). The \emph{consumer surplus never decreases} and, in fact, often increases compared to when the independent seller is the sole seller (\cref{prop:CS_increases}).
\end{enumerate}

\subsection{Related literature}

We draw upon many classical ideas in economics to analyze the modern problem of competition in online retail marketplaces. 

\paragraph{Non-simultaneous duopolies.}
Early work on duopolies consider \emph{simultaneous} actions with a \emph{single decision variable} (either price or quantity). In a Cournot duopoly \citep{cournot1897researches}, sellers $A$ and $B$ simultaneously choose their quantities $q_A$ and $q_B$ and each face price $p = f(q_A + q_B)$. In a Bertrand duopoly \citep{bertrand1883review}, sellers simultaneously choose their prices $p_A$ and $p_B$, then the lower-priced seller $i \in \{A,B\}$ gets demand $D(p_i)$ and the other gets zero demand. When $p_A=p_B$, each seller gets demand $D(p_i)/2$.

The study of non-simultaneous duopolies began with
\citet{stackelberg1934marktform}, who analyzes a \emph{sequential} version of the Cournot duopoly where one seller is a ``leader,'' who sets their quantity first and the other seller is a ``follower,'' who sets their quantity only after observing the leader's quantity. ``Stackelberg'' is now generally used to refer to games with a leader and a follower. 
Whereas Stackelberg games are two-stage games, \citet{maskin1988theory} consider an infinite-stage game where two sellers take turns choosing prices ad infinitum. They use this model to formalize the phenomenon first presented in \citet{edgeworth1925pure}
(now called the ``Edgeworth cycle'') that there is no pure-strategy equilibrium in the infinite sequential game. 
\citet{pal1998endogenous} and subsequent work by
\citet{nakamura2009endogenous} and others study the question of when non-simultaneous duopolies arise naturally between a public and a private firm
by modeling timing as an endogenous decision.
In our paper, we analyze a sequential game with two decision variables --- price \emph{and} quantity.

\paragraph{Capacity constraints and rationing rules.} When price and quantity are both decision variables, we are required to think carefully about how demand is split between sellers. In other words, to analyze games in which players are only able to sell up to an exogenous or endogenously chosen inventory level (i.e., capacity-constrained games), we must specify rationing rules. 
\citet{shubik1959strategy} and \citet{levitan1978duopoly} propose and solve the Bertrand-Edgeworth game, which is the capacity-constrained price competition game.
\citet{kreps1983quantity} consider a two-stage duopoly where both firms choose their inventory in the first stage then set their price in the second stage. They make the intensity rationing assumption (see their Equation 2) and show that under certain conditions, the equilibrium of the two-stage game is the same as in the Cournot duopoly. Later work by \citet{davidson1986long} shows that this result relies crucially on the intensity rationing assumption and does not hold under other rationing assumptions, in particular, proportional rationing. At the proportional rationing equilibrium, firms choose larger capacities and receive less profit than under intensity rationing. With this in mind we consider both intensity and proportional rationing rules in this paper. 
See Section 14.2 of \citet{rasmusen1989games} or Section 5.3.1 of \citet{tirole1988theory} for an overview of rationing rules. Closely related to our work is
\citet{boyer1987being}, who study sequential duopolies with both price and quantity as decision variables, and 
\citet{boyer1989endogenous} and \citet{yousefimanesh2023strategic}, who analyze a similar setting but with imperfectly substitutable goods.

\paragraph{Mixed oligopolies.} In a mixed oligopoly,  welfare-maximizing public firm(s) compete with profit-maximizing private firm(s) \citep{cremer1989public, defraja1990game}.
Our model represents an intermediate between a standard private-private duopoly and a mixed public-private duopoly. The independent seller is a standard private firm, but the marketplace operator can be viewed as somewhere between private and public: ``private'' because they are utility-maximizing rather than welfare-maximizing but close to ``public'' because their utility function internalizes consumer benefit (via the $k$ term, as will be described in Section \ref{sec:preliminaries}).

\paragraph{Marketplaces.}
Whereas the aforementioned related works are all from classical economics literature, the applied motivation of our work most closely relates to recent work that models game dynamics in online marketplaces. 
There is a line of work that models compatibility between buyers and sellers via networks \citep{birge2021optimal, banerjee2017segmenting, damicowong2024disrupting, eden2023platform}, which also relates more broadly to two-sided markets and platform economics \citep{caillaud2003chicken, armstrong2006competition, rochet2003platform}.
Unlike these works, 
we model the marketplace operator itself as a seller, similar to \citet{pabari2025shared}. They analyze the equilibrium of the simultaneous price-competition game and consider the problem of setting the referral fee.
Also related is \citet{shopova2023private}, who studies the effect of a marketplace operator introducing a lower quality private-label product in a price duopoly setting.

Despite the widespread dominance of e-commerce marketplaces and platform-driven competition, the strategic interaction between marketplace operators and independent sellers has not been well studied. Our work directly addresses this need by developing an asymmetric price-quantity duopoly model that captures the core dynamics of platform competition. Unlike prior work, we explicitly analyze multiple rationing rules to reflect a range of real-world scenarios, providing a framework that is not only theoretically rigorous but also practically relevant for designing pricing and inventory strategies in marketplaces.

\section{Model}

We study a game with two players: a marketplace operator ($\MO$) and an independent seller ($\IS$).
Each player must choose the price they will set and the amount of inventory they will order.
$\MO$ and $\IS$ can each produce a product at costs $c_{\MO} \geq 0$ and $c_{\IS} \geq 0$ per unit respectively. For now, we assume that customers view the products produced by each seller as perfect substitutes (e.g., both players sell light bulbs), although we will consider an extension to imperfect substitutability in Section \ref{sec:imperfect_substitutability}. $\MO$ operates the market in which the players sell their products, so $\IS$ pays $\MO$ a fraction $\alpha \in [0,1]$ of their revenue per unit as a referral fee. Furthermore, $\MO$ realizes benefit $k \geq 0$ for every unit sold to account for the customer's positive experience, which contributes to the health of the marketplace (because a happy customer is more likely to make future purchases). 
More concretely, $k$ can be thought of as the increase in future profit if the customer is able to buy the good compared to if they are not able. 
The number of units each player can sell is the minimum of their inventory and their demand, and any unsold units have zero value.\footnote{For simplicity, our formulation assumes that leftover inventory has zero salvage value; in reality, products that have steady demand generally have positive salvage value because excess inventory can be sold in the next time period. Relaxing this assumption is a natural topic for future work.} Players are allowed to stock and sell fractional units. We consider a sequential game in which $\MO$ is the first mover, to reflect the reality that large platforms are less agile due to their complex interlocking operations and have to commit to decisions well in advance in order for business to run smoothly; third-party sellers, on the other hand, can generally be more reactive. We now provide a formal definition of the game.

\begin{game}[label=game, colback=white, colframe=gray]

\emph{Parameters:} referral fee $\alpha$, consumer experience term $k$, costs $c_{\MO}$ and $c_{\IS}$
\begin{enumerate}[noitemsep] 
    \item The marketplace operator ($\MO$) chooses their price $p_{\MO} \geq 0$ and quantity $q_{\MO} \geq 0$.
    \item The independent seller ($\IS$) observes $p_{\MO}$ and $q_{\MO}$, then sets their price $p_{\IS} \geq 0$ and quantity $q_{\IS} \geq 0$. 
    \item Both players realize their respective utilities
    \begin{align}
        u_{\MO}(p_{\MO}, q_{\MO}, p_{\IS}, q_{\IS}) &= (p_{\MO} + k) \min(q_{\MO}, D_{\MO}) + (\alpha p_{\IS} + k) \min(q_{\IS}, D_{\IS}) - c_{\MO} q_{\MO}, \\
        u_{\IS}(p_{\MO}, q_{\MO}, p_{\IS}, q_{\IS}) &= (1-\alpha)p_{\IS} \min(q_{\IS}, D_{\IS}) - c_{\IS} q_{\IS},
    \end{align}
    where $D_{\MO}$ and $D_{\IS}$, which we will describe below, denote the respective demands for $\MO$ and $\IS$ and are functions of $p_{\MO}$, $q_{\MO}$, $p_{\IS}$, and $q_{\IS}$. 
\end{enumerate}
\end{game}

\emph{Demand functions.}
Before we can study the game, we first need to specify each player's demand function, which can be thought of in terms of ``original'' and ``residual'' demand functions. The player $j \in \{\MO, \IS\}$ that sets the lower price faces the original demand function $Q(p_j)$ and the other seller $i$ faces the residual demand function $R(p_i; q_j, p_j)$. If prices are equal, the tie is broken in favor of $\IS$.
The independent seller's demand function is thus 
\begin{align*}
    D_{\IS}(p_{\IS}; q_{\MO}, p_{\MO}) = \begin{cases}
        Q(p_{\IS}) & \text{if } p_{\IS} \leq p_{\MO}, \\
        R(p_{\IS}; q_{\MO}, p_{\MO})  & \text{if } p_{\IS} > p_{\MO},
    \end{cases}
\end{align*}
and the marketplace operator's demand function $D_{\MO}(p_{\MO}; q_{\IS}, p_{\IS})$ is the same but with $\IS$ and $\MO$ subscripts switched and with a $<$ in the first line and a $\geq$ in the second line.

For mathematical concreteness, we assume a linear demand function with unit slope\footnote{The unit slope assumption is without loss of generality within the class of linear demand functions: any downwards-sloping linear function can be rewritten to have unit slope by simply redefining what constitutes one unit of quantity. For example, if raising the price of gas by \$1 causes demand to decrease by 1000 gallons, we can just define one unit of demand to be 1000 gallons.} and intercept $\theta \geq 0$. $\theta$ corresponds to the maximum willingness to pay of any customer and is also the number of units that would be demanded if the price were zero.

\begin{assumption}[Linear demand] \label{assumption:linear_demand}
The quantity demanded at price $p$ is 
\begin{align*}
    Q(p) = \begin{cases}
        \theta - p & \text{for $0 \leq p \leq \theta$} \\
        0  & \text{for $p>\theta$}.
    \end{cases}
\end{align*}
\end{assumption}

\begin{wrapfigure}{r}{0.3\textwidth}
    \centering
    \vspace{-3cm}
    \begin{tikzpicture}[scale=0.8]
    \draw[->] (0,0) -- (4.5,0);
    \draw[->] (0,0) -- (0,4.5);

    \node[below] at (2.25,-0.5) {Quantity};
    \node[rotate=90] at (-0.8,2.25) {Price};

    \def\thetaval{4}
    \def\qA{1}

    \draw[thick, domain=0:\thetaval] plot (\x,{\thetaval - \x});

    \draw[blue, very thick, dashed, domain=0:{\thetaval - \qA}] plot (\x,{\thetaval - \qA - \x});

    \draw [<->] (2.6, \thetaval - \qA - 2.5) -- (3.4, \thetaval - \qA - 2.5);

    \node at (\thetaval,-0.35) {$\theta$};
    \node at (-0.3,\thetaval) {$\theta$};
    \node at (2,2.8) {$Q(p_{\IS})$};
    \node[blue] at (1.1,2 - \qA) {$R(p_{\IS})$};
    \node at (2.9, \thetaval - \qA - 2.3) {$q_{\MO}$};
    \node[white] at (3,-.35) {};
\end{tikzpicture}

    
    


    

    \caption{Visualization of $R(p_{\IS})$ for $\theta=4$, $q_{\MO}=1$, and $p_{\MO}=2$ under intensity rationing.}
    \label{fig:R_intensity}
\end{wrapfigure}

The residual demand function is determined by a \emph{rationing rule} that determines which customers are able to buy at the lower price and which remain for the higher-priced seller.   
There exists a range of reasonable rationing rules and associated residual demand functions and they have varying degrees of favorability for the higher-priced seller (for an extended discussion, see Appendix \ref{sec:rationing_rule_APPENDIX}).
To start, we will focus on analyzing the rationing rule that is least favorable to the higher-priced seller, known as \emph{intensity rationing}. 

\begin{assumption}[Intensity rationing]
    \label{assumption:intensity_rationing} Whenever demand exceeds supply, the customers with the \emph{highest valuation} are the ones who are able to buy. Thus, if the lower-priced seller $j$  sets price $p_j$ and has inventory $q_j \leq Q(p_j)$, then the residual demand function for the higher-priced seller $i$ is
    $$R(p_i) = Q(p_i) - q_j.$$
\end{assumption}

Figure \ref{fig:R_intensity} illustrates the residual demand function $R$ under intensity rationing relative to the original demand function $Q$. We see that the residual demand function for the higher-priced seller is simply the original demand function shifted down by the number of units sold at the lower price. 
We assume intensity rationing throughout the paper except for Section \ref{sec:proportional_rationing}, where we explore the implications of proportional rationing.

\section{Best Response and Competitive Pricing} \label{sec:inducing_competitive_pricing}

At a technical level, this section describes the independent seller's best response as a function of the marketplace operator's price $p_{\MO}$ and inventory level $q_{\MO}$. At a practical level, these results provide useful insights on how marketplace operators can, by setting their own price and inventory appropriately, induce sellers in their marketplace to price competitively. 



\subsection{Preliminaries} \label{sec:preliminaries}

We now define two key prices, which we will use when defining $\IS$'s best response function. 

\begin{definition} The independent seller's \emph{break-even price} $p_0$ is the price at which they get zero utility from selling the good:
    $$p_0 = \frac{c_{\IS}}{1-\alpha}.$$
\end{definition}
This definition follows directly  
upon recalling that $\IS$'s utility (i.e., net profit) from selling one unit of the good at price $p$ after paying the referral fee $\alpha$ is $(1-\alpha)p - c_{\IS}$.
We assume that given the choice between selling a positive quantity at $p_0$ vs. not selling at all, $\IS$ will choose to sell at $p_0$. 

\begin{definition} The independent seller's \emph{optimal sole-seller price} $\psole$ is the price the independent seller would set when the marketplace operator is not a seller ($q_{\MO}=0$): 
    \begin{align*}
        \psole &= \arg \max_p ((1-\alpha)p - c_{\IS}) Q(p) \\
        &= \begin{cases}
        \infty & \text{if } p_0 > \theta,  \\
         \frac{1}{2}(p_0+\theta) & \text{otherwise}, 
    \end{cases}
    \end{align*}
where $\infty$ can be interpreted to be any price that is above $\theta$ or, equivalently, any price $p$ at which $Q(p) = 0$. 
\end{definition}

When we say \emph{induce competitive pricing}, we refer to actions that result in the independent seller setting a price strictly less than $\psole$.

We see in the previous displayed equation that if $c_{\IS}$ and $\alpha$ are so high that $p_0 > \theta$, the independent seller will not sell. This results in a simple, near trivial, solution of the game, which we describe in the following proposition.

\begin{proposition}[Trivial solution for large $p_0$]
    Whenever $p_0 > \theta$, the equilibrium solution is that the independent seller does not sell ($q_{\IS}=0$) and the marketplace operator sets price $\psoleA$ and quantity $Q(\psoleA)$, where
    \begin{align*}
        \psoleA &= \arg \max_p (p + k - c_{\MO}) Q(p) \\
        &= \begin{cases}
            \infty &\text{if } c_{\MO} > \theta+k,  \\
            \frac{1}{2}(c_{\MO} - k + \theta) & \text{otherwise.}
        \end{cases}
    \end{align*}
\end{proposition}

Since the condition $p_0 > \theta$ does not generally occur in the real world (this would mean, e.g., that an independent seller's manufacturing cost is close to customers' maximum willingness to pay) and is also uninteresting, we will focus on analyzing the $p_0 \leq \theta$ case in the rest of the paper.

\subsection{$\IS$'s best response}
\label{sec:best_response}

Although the independent seller action space is comprised of two variables --- price and quantity --- the problem of determining $\IS$'s best action reduces to simply determining their best price.
The following proposition,
which is similar to Proposition 1 of \citet{yousefimanesh2023strategic}, formalizes the intuition that the independent seller should sell as much as they can at their chosen price, so they should set their quantity to exactly meet their demand. This follows from the fact that, unlike the marketplace operator, $\IS$ only benefits from directly selling the product. Note that even though we state the result in a way that implies the maximizing $q_{\IS}$ can be non-unique, the only case in which there is no unique optimal $q_{\IS}$ given $p_{\IS}$ is if $p_{\IS}=p_0$. In that case, any $q_{\IS} \leq  D_{\IS}(p_{\IS}; q_{\MO}, p_{\MO})$ is optimal. In all other cases, $q_{\IS} = D_{\IS}(p_{\IS}; q_{\MO}, p_{\MO})$ is the unique optimal quantity.

\begin{proposition}[$\IS$ meets demand] \label{prop:optimal_qB}
    Given fixed $p_{\IS}$, $q_{\MO}$, and $p_{\MO}$, if there exists $q_{\IS}$ such that $u_{\IS}(p_{\MO}, q_{\MO}, p_{\IS}$, $ q_{\IS}) \geq 0$ (which happens if and only if $p_{\IS} \geq p_0$), then we must have
    \begin{align*}
        D_{\IS}(p_{\IS}; q_{\MO}, p_{\MO}) \in \arg\max_{q_{\IS}} u_{\IS}(p_{\MO}, q_{\MO}, p_{\IS}, q_{\IS}),
    \end{align*}
    i.e., it is utility-maximizing for the independent seller to set $q_{\IS}$ equal to their demand.
\end{proposition}

We prove Proposition \ref{prop:optimal_qB} in Appendix \ref{sec:auxiliary_proofs_APPENDIX}.
Given Proposition \ref{prop:optimal_qB}, we can focus on deriving $\IS$'s best response in terms of their price $p_{\IS}$ with the assumption that they will correspondingly set their inventory as $q_{\IS} = D_{\IS}(p_{\IS}; q_{\MO}, p_{\MO})$. We can thus remove the dependence of $\IS$'s utility on $q_{\IS}$ by plugging in the best response condition, as follows:
\begin{align}
    u_{\IS}(p_{\IS}, p_{\MO}, q_{\MO}) = [(1-\alpha)p_{\IS}-c_{\IS}]D_{\IS}(p_{\IS}; q_{\MO}, p_{\MO}).
\end{align}
We adopt the convention that $\IS$ sets $p_{\IS} = \infty$ when they wish to abstain (by setting $q_{\IS} = 0$).
Having established that $\IS$'s decision reduces to choosing a value of $p_{\IS}$, we now describe the strategies $\IS$ can choose from.

\begin{figure}
    \centering
    \includegraphics[width=0.4\linewidth]{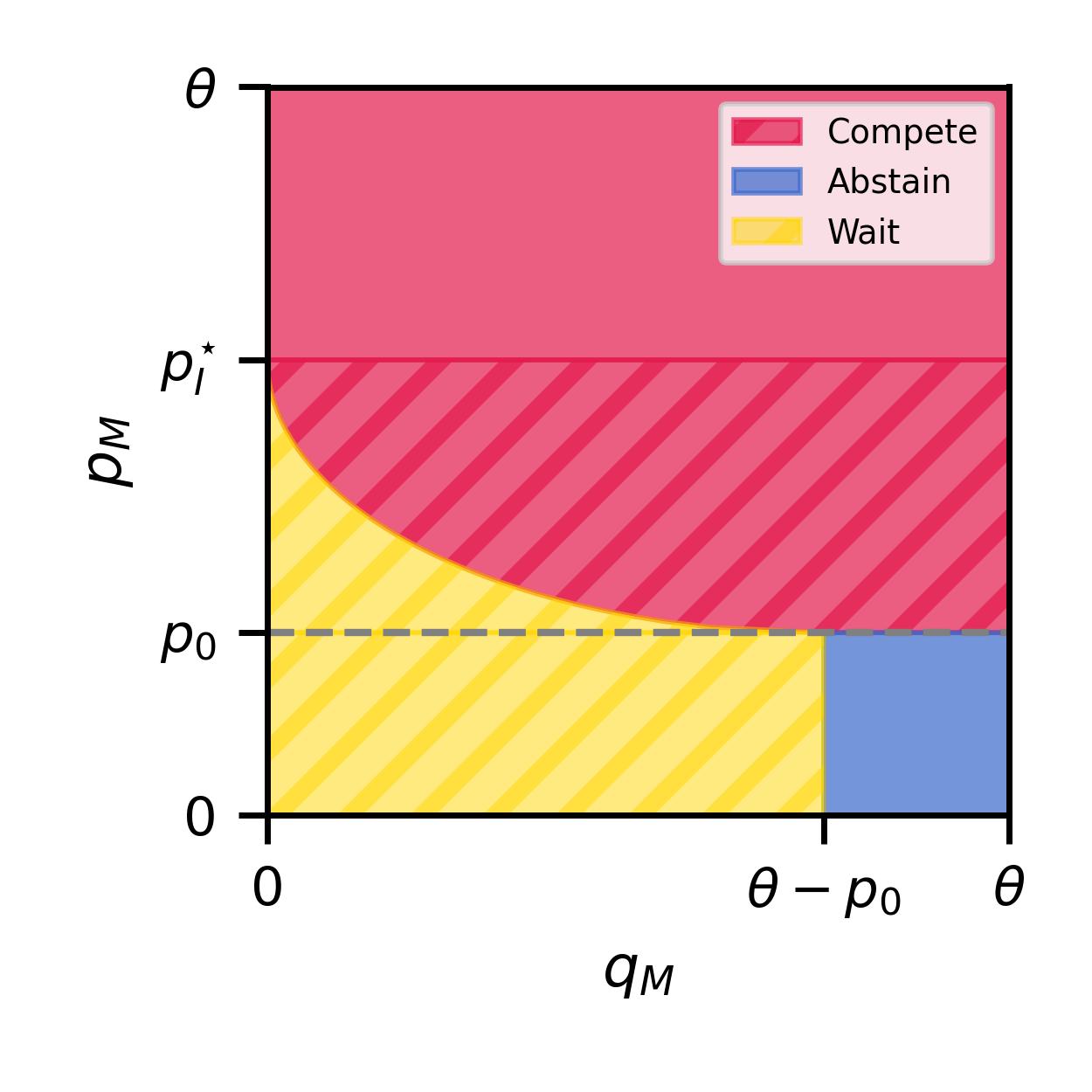}
    \vspace{-20pt}
    \caption{Visualization of $\IS$'s best response to different $q_{\MO}$ (x-axis) and $p_{\MO}$ (y-axis) combinations. Stripes are used to denote regions where $\IS$ is de-monopolized. Game parameters are set to $\theta=10$, $c_{\IS}=2$, and $\alpha=0.2$ (the value of $c_{\MO}$ does not affect $\IS$'s best response).}
    \label{fig:best_response_intensity}
\end{figure}

\begin{tcolorbox}[colback=white, colframe=gray] 
\emph{$\IS$'s possible strategies:}
\begin{enumerate}[noitemsep]
    \item \textsc{Compete} by setting $p_{\IS} \leq p_{\MO}$ (and thus face the original demand function $Q$).
    \item \textsc{Wait (it out)} by setting $p_{\IS} > p_{\MO}$ (and thus face the residual demand function $R$).
    \item \textsc{Abstain} by setting $p_{\IS} = \infty$.
\end{enumerate}
\end{tcolorbox}

We assume that when $\IS$ is indifferent between competing and not competing, they choose to compete. 
Note that just because $\IS$ competes (by setting $p_{\IS} \leq p_{\MO}$) does not necessarily mean that they are setting a competitive price (which requires $p_{\IS} < \psole$). Furthermore, $\IS$ competing is not a necessary condition for them to set a competitive price; they can wait and still set a price less than $\psole$. To allow us to better distinguish between these cases, we will use \emph{monopolistic} to refer to an outcome in which $\IS$ sets their sole-seller price $\psole$ and \emph{de-monopolized} to refer to an outcome in which $\IS$ sets a price strictly lower than $\psole$. We now describe the independent seller's best response function.


\begin{lemma}[Best response under intensity rationing]
\label{lemma:best_response_intensity} 
\begin{samepage} 
Define the following reference values
\[
q^{\dagger}(p_{\MO}) = \theta - p_0 - 2 \sqrt{(p_{\MO} - p_0)(\theta-p_{\MO})},
\qquad
q^{\ddagger}(p_{\MO}) = \theta - p_0,
\qquad
p_{\IS}^W = \psole - \frac{q_{\MO}}{2}.
\]
The independent seller’s optimal price given $(p_{\MO},q_{\MO})$ is
\begin{align*}
p_{\IS}^{\BR}(p_{\MO},q_{\MO})
  =\left\{
     \begin{array}{l@{\;\;}l@{\;\;}r}
       \psole      &
       \textnormal{if } p_{\MO}\ge \psole                                      &
       \quad \text{(monopolistic \textsc{Compete})},\\[4pt]
       p_{\MO}     &
       \textnormal{if } p_{\MO}\in[p_0,\psole)\ 
       \textnormal{and } q_{\MO}\ge q^{\dagger}(p_{\MO})                       &
       \quad \text{(de-monopolized \textsc{Compete})},\\[4pt]
       p_{\IS}^{W} &
       \textnormal{if } p_{\MO}\in[p_0,\psole)\ 
       \textnormal{and } q_{\MO}< q^{\dagger}(p_{\MO})                         &
       \quad \text{(de-monopolized \textsc{Wait})},\\[4pt]
       p_{\IS}^{W} &
       \textnormal{if } p_{\MO}< p_0\ 
       \textnormal{and } q_{\MO}< q^{\ddagger}(p_{\MO})                        &
       \quad \text{(de-monopolized \textsc{Wait})},\\[4pt]
       \infty      &
       \textnormal{if } p_{\MO}< p_0\ 
       \textnormal{and } q_{\MO}\ge q^{\ddagger}(p_{\MO})                      &
       \quad \text{(\textsc{Abstain})},
     \end{array}
   \right.
\end{align*}
and their corresponding optimal inventory is $q_{\IS}^{\BR}(p_{\MO},q_{\MO}) = D_{\IS}(p_{\IS}^{\BR}(p_{\MO},q_{\MO}); q_{\MO}, p_{\MO})$.
\end{samepage}
\end{lemma}

This lemma follows from results in Appendix \ref{sec:best_response_proofs_APPENDIX}.
We visualize
Lemma \ref{lemma:best_response_intensity} in Figure 
\ref{fig:best_response_intensity}.

\section{When Should the Marketplace Operator Induce Competition?}

In the previous section, we derived the independent seller's best response given the marketplace operator's price and inventory. Armed with this,
we are now ready to determine the marketplace operator's optimal strategy, i.e., to derive the equilibrium of our game.
Let $\pBR(p_{\MO},q_{\MO})$ and $\qBR(p_{\MO},q_{\MO})$  denote $\IS$'s best response to $p_{\MO}$ and $q_{\MO}$, as described in \Cref{lemma:best_response_intensity}.
Assuming $\IS$ best responds, $\MO$'s utility for playing $(p_{\MO}, q_{\MO})$ for $q_{\MO} \leq Q(p_{\MO})$
can be written as 
\begin{align} \label{eq:u_{MO}}
    u_{\MO}(p_{\MO}, q_{\MO}) =
    \begin{cases}
        u_{\MO}^{\IS\mathrm{comp}}(p_{\MO},q_{\MO})&\text{if } \pBR(p_{\MO},q_{\MO}) \leq p_{\MO}, \\
         u_{\MO}^{\IS\mathrm{wait/abs}}(p_{\MO},q_{\MO}) &\text{if } \pBR(p_{\MO},q_{\MO}) > p_{\MO}, 
    \end{cases}
\end{align}
\begin{align*}
    &\text{where }  \qquad  \qquad
u_{\MO}^{\IS\mathrm{comp}}(p_{\MO},q_{\MO}) =\big(\alpha\pBR(p_{\MO},q_{\MO}) +k\big)  Q(\pBR(p_{\MO}, q_{\MO})) - c_{\MO} q_{\MO} \\
&\text{and } \qquad  \qquad
    u_{\MO}^{\IS\mathrm{wait/abs}}(p_{\MO},q_{\MO}) = (p_{\MO} - c_{\MO} + k) q_{\MO} + \big(\alpha\pBR(p_{\MO},q_{\MO}) +k\big) R(\pBR(p_{\MO}, q_{\MO}); q_{\MO}, p_{\MO}).
\end{align*}
We split $\MO$'s strategy space into a few named partitions.
\begin{tcolorbox}[colback=white, colframe=gray]
\emph{$\MO$'s possible strategies:}
\begin{enumerate}[noitemsep]
    \item \textsc{Induce Abstain} by setting  $p_{\MO}$ low enough and $q_{\MO}$ high enough so that $\IS$'s best response is to abstain from entering the market (i.e., $\IS$ sets $q_{\IS} = 0$). 
    \item \textsc{Induce Compete} by setting a moderate $p_{\MO}$ and $q_{\MO}$ high enough that $\IS$'s best response is to compete and set $p_{\IS} \leq p_{\MO}$.
    \item \textsc{Induce Wait} by setting a low $p_{\MO}$ and $q_{\MO}$ low enough that $\IS$'s best response is to wait and set $p_{\IS} > p_{\MO}$. This partition also includes \textsc{($\MO$) Abstains}, which refers to when $\MO$ does not enter the market ($q_{\MO} = 0$).
\end{enumerate}
\end{tcolorbox}

For a fixed price $p_{\MO}$, we can identify the optimal quantity $q_{\MO}$ that falls within each of the three strategy spaces and then compare the resulting utilities to determine the overall optimal quantity, yielding the following lemma.

\begin{lemma}[Optimal $q_{\MO}$ given $p_{\MO}$]
\label{lemma:qstar_A} 
Define $\qthresheps(p_{\MO}) := \qthresh(p_{\MO})-\epsilon$ (for infinitely small $\epsilon >0$) to be the largest inventory the marketplace operator can acquire while still inducing the independent seller to wait.
The following is an optimal ($u_{\MO}$-maximizing) inventory $q_{\MO}$ given a fixed price $p_{\MO}$:
\begin{align*}
    q^*_{\MO}(p_{\MO}) = \begin{cases}
        0 & \text{if } p_{\MO} \geq \psole, \\
        \arg\max_{q \in \{0, \qthresheps(p_{\MO}), \qthresh(p_{\MO})\}} 
        u_{\MO}(p_{\MO}, q) & \text{if } p_{\MO} \in [p_0, \psole), \\
        \arg\max_{q \in \{0, Q(p_{\MO})\}} 
        u_{\MO}(p_{\MO}, q) & \text{if } p_{\MO} < p_0.
    \end{cases}
\end{align*}
The optimal quantity is unique except for when there are ties in the argmax.
\end{lemma}
 
Lemma \ref{lemma:qstar_A} follows from results in Appendix \ref{sec:equilibrium_APPENDIX}. Using this lemma, we can now write out the equilibrium strategies for each player.

\begin{theorem}[Equilibrium] 
\label{theorem:equilibrium_unconstrained_game}
The following constitutes the subgame-perfect Nash equilibrium of the game:
\begin{enumerate}
    \item $\MO$ chooses price $p_{\MO}^* = \arg\max_{p_{\MO}} u_{\MO}(p_{\MO}, q_{\MO}^*(p_{\MO}))$
    and quantity $q_{\MO}^*(p_{\MO}^*)$.
    \item $\IS$ chooses price $\pBR(p_{\MO}^*, q_{\MO}^*(p_{\MO}^*))$ and quantity  $\qBR(p_{\MO}^*, q_{\MO}^*(p_{\MO}^*))$.
\end{enumerate}
\end{theorem}

\begin{remark} \label{remark:efficient_computation_of_equilibrium}
    $p_{\MO}^*$ can be identified by solving at most three single-variable optimization problems and comparing the resulting utilities. 
    Reading off from Lemma \ref{lemma:qstar_A}, we see that we simply have to compare the optimal utilities for 
    (1) any $p_{\MO}$ and $q_{\MO} = 0$; (2) $p_{\MO} \in [p_0, \psole)$ and $q_{\MO} =\qthresh(p_{\MO})$;
    (3) $p_{\MO} \in [p_0, \psole)$ and $q_{\MO} = \qthresheps(p_{\MO})$; and (4) $p_{\MO} \in [0, p_0)$ and $q_{\MO} = Q(p_{\MO})$.
    $\MO$'s utility for (1) is 
    $u_{\MO} = (\alpha \psole + k) Q(\psole)$, which is constant in $p_{\MO}$. The optimal prices for (2),(3), and (4) can be easily identified using a numerical solver. More details are provided in Appendix \ref{sec:equilibrium_APPENDIX}.
\end{remark}

Recall that our initial motivation is to understand when it is optimal for the marketplace operator to induce competition. 
Is it ever optimal for $\MO$ to induce $\IS$ to compete? Yes. Is it always optimal? No. 

\begin{remark}
    There exist game parameters that make each of the marketplace operator's possible strategies optimal (\textsc{Induce Abstain}, \textsc{Induce Compete}, and \textsc{Induce Wait}). 
\end{remark}

This is visualized in Figure \ref{fig:cA_cB_phase_diagram}. Here, and in all subsequent plots, we use $\theta=10$. We mark $\theta(1-\alpha)$ on the $c_{\IS}$ axis because this is a critical threshold at which $p_0 = \theta$; above this threshold, $\IS$ always abstains regardless of what $\MO$ does, as described in Section \ref{sec:preliminaries}.
In Figure \ref{fig:cA_cB_phase_diagram}, we observe that when $c_{\MO}$ is high and $c_{\IS}$ is low, $\MO$ will prefer to let $\IS$ sell, so they induce $\IS$ to compete. When $c_{\IS}$ is relatively large, $\MO$ will prefer to sell themselves and $\IS$ will be induced to abstain. There is also a small region where $\MO$ and $\IS$ have roughly equal costs in which $\MO$ finds it optimal to induce $\IS$ to wait. 
Relating back to the concept of de-monopolization, note that the red and yellow regions correspond to regions where $\IS$ is de-monopolized at equilibrium. 
Figure \ref{fig:cA_cB_prices_and_quantities} visualizes the corresponding prices and quantities. We observe that for the ($c_{\MO}$, $c_{\IS}$) combinations that fall in the red ``$\IS$ competes'' region in Figure \ref{fig:cA_cB_phase_diagram}, 
$\MO$ induces compete by buying only a small quantity, which is enough to induce $\IS$ to set a slightly lower price than they would otherwise. For example, when $c_{\MO} = 3$ and $c_{\IS}=1$, $\MO$ will set $p_{\MO} = 4.39$ and $q_{\MO}=0.35$, which causes $\IS$ to respond by setting $p_{\IS}=4.39$ and $q_{\IS}=5.61$.
Note that 4.39 is lower than $\IS$'s sole-seller price of $\psole = 0.5 (\frac{1}{1-0.2} + 10) = 5.625$.

Using the derived equilibrium, we can also answer practical questions such as \emph{(i) how should the marketplace operator set $\alpha$?} and \emph{(ii) should the marketplace operator behave differently when selling a product with a larger impact on customer experience (large $k$)?} Answers to these questions are provided in Appendix \ref{sec:applying_equilibrium_APPENDIX}.

\begin{figure}[h!]
\centering
\begin{subfigure}{.45\textwidth}
  \centering
\includegraphics[width=0.7\textwidth]{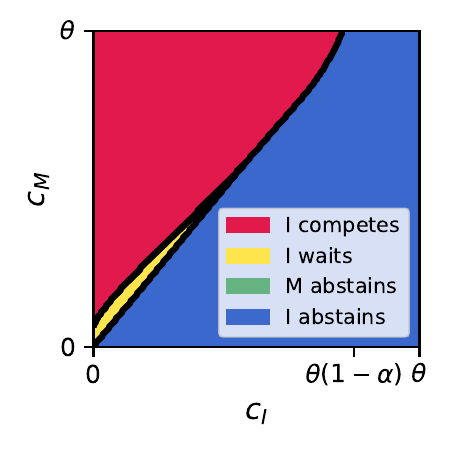}
    \vspace{-5pt}
  \caption{High-level equilibrium characterization}
  \label{fig:cA_cB_phase_diagram}
\end{subfigure}%
\begin{subfigure}{.45\textwidth}
  \centering
  \includegraphics[width=0.9\textwidth]{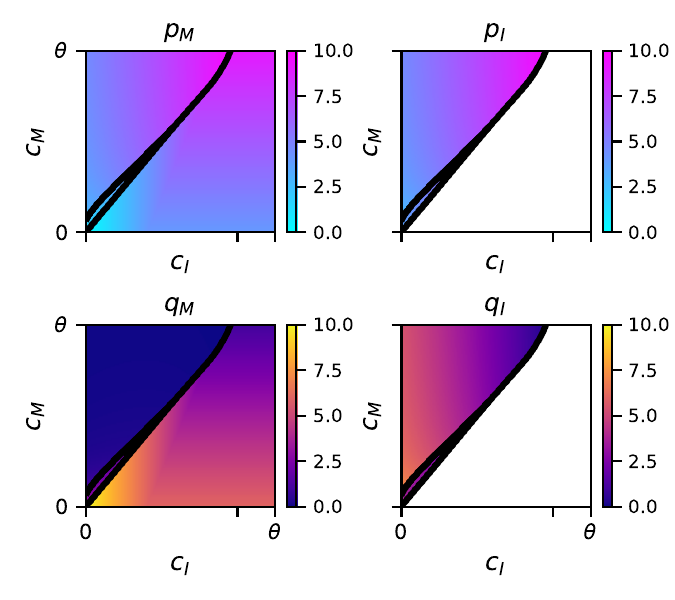}
  \vspace{-5pt}
  \caption{Equilibrium prices and quantities}
  \label{fig:cA_cB_prices_and_quantities}
\end{subfigure}

\vspace{-5pt} 
    \caption{Equilibrium characterization for each combination ($c_{\MO}, c_{\IS}$) when $\alpha=0.2$ and $k=2$. 
    In the left plot,
    red denotes cost combinations where it is optimal for $\MO$ to induce $\IS$ to compete; blue regions are where $\MO$ sets a price that causes $\IS$ to abstain, yellow regions are where $\MO$ finds it optimal to induce $\IS$ to wait; and green is used to denote regions where $\MO$ chooses to abstain, but no such regions exist in the plots. The set of plots on the right show the corresponding equilibrium prices and quantities.} 
\end{figure}

\section{Implications for Consumers} \label{sec:welfare}

A key question is how consumers are affected by the participation of the marketplace operator in their own market. This motivates us to analyze consumer surplus, which is defined as the sum of excess value that customers who are able to buy the good obtain from purchasing it. Let $P$ be the inverse demand function, where $P(q)$ is the price at which the demand $Q$ is exactly $q$.  Since we assume the demand function is $Q(p) = \theta - p$ for $p \in [0,\theta]$ (Assumption \ref{assumption:linear_demand}), the inverse demand function is $P(q) = \theta - q$. Thus, we have $Q(t) = P(t)$ for $t \in [0,\theta]$. When $q \le Q(p)$ units of the product are available at price $p$, the consumer surplus is
\begin{align} \label{eq:CS_one_seller}
    \CS = \int_{0}^{q} \big(P(t) - p\big)\,dt = \int_{0}^{q} \big(Q(t) - p\big)\,dt.
\end{align}
When the product is available at two prices, $p_i < p_j$, with $q_i \le Q(p_i)$ units at price $p_i$ and $q_j = Q(p_j) - q_i$ units at price $p_j$ (i.e., the higher-priced seller exactly meets the residual demand), the consumer surplus is
\begin{align} \label{eq:CS_two_sellers}
    \CS = \int_{0}^{q_i} \big(P(t) - p_i\big)\,dt 
    + \int_{q_i}^{Q(p_i)} \big(P(t) - p_j\big)\,dt
    =
    \int_{0}^{q_i} \big(Q(t) - p_i\big)\,dt 
    + \int_{q_i}^{Q(p_i)} \big(Q(t) - p_j\big)\,dt.
\end{align}

We use $\CS(p_{\MO}, q_{\MO}, p_{\IS}, q_{\IS})$ to denote the consumer surplus when the marketplace operator plays $(p_{\MO}, q_{\MO})$ and the independent seller plays $(p_{\IS}, q_{\IS})$.
Let
$\CSsole := \CS(\infty, 0, \psole, Q(\psole))$ be the consumer surplus when the marketplace operator does not participate as a seller and let $\uISsole := u_{\IS}(0, 0, \psole, Q(\psole))$ be the corresponding independent seller utility. The following lemma tells us that the entry of the marketplace operator has the pro-consumer effect of transferring surplus from the independent seller to consumers. 

\begin{lemma}[Surplus transfer] \label{lemma:surplus_transfer}
For any $p_{\MO}$ and $q_{\MO}$ (including the equilibrium values), as long as the independent seller best responds, any decrease in the independent seller's surplus caused by the marketplace operator's entry is made up for by an increase in the consumer surplus, that is,
$$- \Delta \PS \leq \Delta \CS,$$
where $\Delta \PS = u_{\IS}\big(p_{\MO}, q_{\MO}, \pBR(p_{\MO}, q_{\MO}), \qBR(p_{\MO}, q_{\MO})\big) - \uISsole$ and $\Delta \CS = \CS\big(p_{\MO}, q_{\MO}, \pBR(p_{\MO}, q_{\MO}), \qBR(p_{\MO}, q_{\MO})\big) - \CSsole$.
\end{lemma}

The following related lemma states that the entry of marketplace operator increases consumer surplus.

\begin{lemma}[Consumer surplus increases] \label{prop:CS_increases}
    For any $p_{\MO}$ and $q_{\MO}$ (including the equilibrium values), as long as the independent seller best responds, the consumer surplus will be at least as high as if the marketplace operator did not participate as a seller: 
    $$\CS\big(p_{\MO}, q_{\MO}, \pBR(p_{\MO}, q_{\MO}), \qBR(p_{\MO}, q_{\MO})\big) \geq \CSsole \qquad \forall p_{\MO}, q_{\MO}. $$
    The inequality is strict whenever $p_{\MO} < \psole$. 
\end{lemma}

Proofs of these results, as well as additional investigations into total welfare, are provided in Appendix \ref{sec:welfare_APPENDIX}.

\section{Robustness of Equilibrium Analysis}

We have thus far focused on the intensity rationing assumption applied to perfectly substitutable goods. In this section we extend our analysis to proportional rationing and imperfectly substitutable goods.

\subsection{Proportional rationing} \label{sec:proportional_rationing}

Intensity rationing assumes that when there is a limited supply of the good at a lower price, the customers who are able to buy at the lower price are the ones with the highest valuation. We now consider an alternative demand splitting rule called \emph{proportional rationing} (also known as Beckmann rationing; see \citet{davidson1986long}): Suppose that every customer has the same probability $\rho$ of being assigned to the lower-priced seller, where $\rho$ is such that in expectation the demand that the lower-price seller faces at their price $p_j$ is exactly equal to their supply $q_j$. This yields the residual demand function
    $$R(p_i) = Q(p_i)  \left(1-\frac{q_j}{Q(p_j)}\right)$$
for the higher-priced seller $i$.
This residual demand function is visualized in Figure \ref{fig:R_proportional}. 
As with intensity rationing, we can derive the independent seller's best response function and use this to compute the equilibrium, which we visualize in Figures \ref{fig:best_response_proportional} and \ref{fig:cA_cB_phase_diagram_proportional}. The formal analysis and precise expressions are deferred to Appendix \ref{sec:proofs_APPENDIX}.
For the most part, the results are similar to those under intensity rationing, except the boundaries between regimes have shifted slightly.
One notable difference is that under proportional rationing, $\IS$ is no longer de-monopolized when they wait. Under intensity rationing, $\IS$'s price when waiting is $\psole - q_{\MO}/2$ but under proportional rationing, it is simply their sole-seller price $\psole$.

One limitation of the proportional rationing assumption is that it is not straightforward to consider consumer surplus in this setting; under intensity rationing, consumer surplus is simply the area between the demand function $Q(p)$ and the purchase price, but this is no longer true under proportional rationing. Thinking more carefully about how to compute consumer surplus under proportional rationing, as well as the implications of a participatory $\MO$ in such settings, is a worthwhile direction for future work.

\begin{figure}[h!]
\centering
\begin{subfigure}{.32\textwidth}
  \centering
    \begin{tikzpicture}[scale=0.8]
    \draw[->] (0,0) -- (4.5,0);
    \draw[->] (0,0) -- (0,4.5);

    \node[below] at (2.25,-0.5) {Quantity};
    \node[rotate=90] at (-0.8,2.25) {Price};

    \def\thetaval{4}
    \def\qA{1}

    \draw[thick, domain=0:\thetaval] plot (\x,{\thetaval - \x});

    \draw[blue, very thick, dashed] (0, \thetaval) -- (2, 0);
    

    \node at (\thetaval,-0.35) {$\theta$};
    \node at (-0.3,\thetaval) {$\theta$};
    \node at (2,2.8) {$Q(p_{\IS})$};
    \node[blue] at (.8,2 - \qA) {$R(p_{\IS})$};
    \node[blue] at (2, -0.3) {$\theta \big(1-\frac{q_{\MO}}{Q(p_{\MO})}\big)$};
    \node[white] at (3,-.35) {};
\end{tikzpicture}
  \vspace{-20pt}
  \caption{Residual demand function}
  \label{fig:R_proportional}
\end{subfigure}
\begin{subfigure}{.32\textwidth}
  \centering
\includegraphics[width=\textwidth]{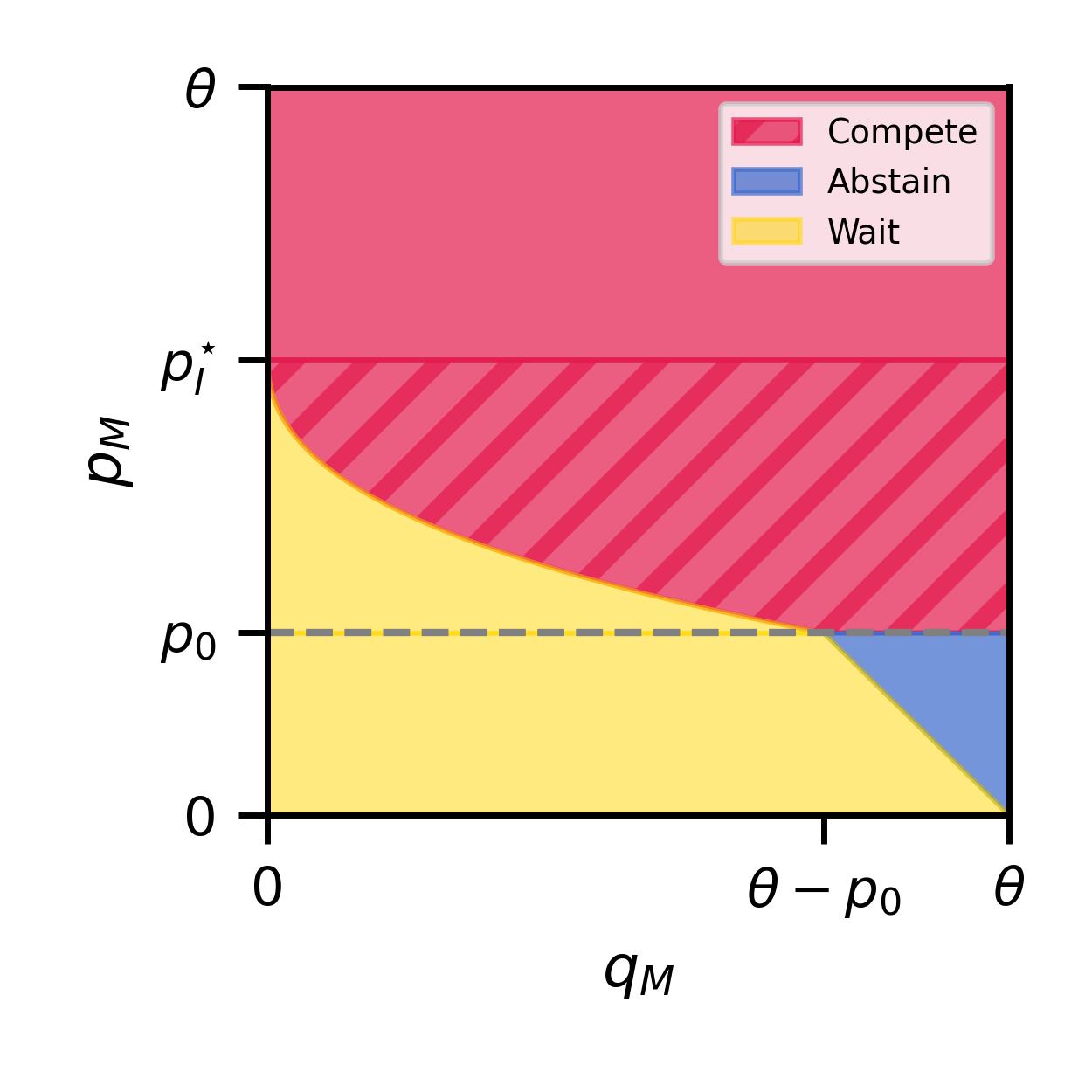}
    \vspace{-20pt}  
  \caption{$\IS$'s best response}
  \label{fig:best_response_proportional}
\end{subfigure}
\begin{subfigure}{.32\textwidth}
  \centering
  \includegraphics[width=\textwidth]{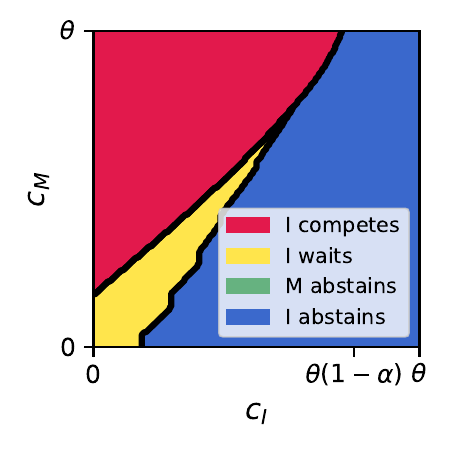}
   \vspace{-20pt} 
  \caption{Equilibrium characterization}
  \label{fig:cA_cB_phase_diagram_proportional}
\end{subfigure}
\vspace{-5pt} 
\caption{Residual demand function, $\IS$'s best response, and equilibrium characterization under \textbf{proportional rationing}. The game parameters in Figures \ref{fig:best_response_proportional} and \ref{fig:cA_cB_phase_diagram_proportional} are set to the same values as Figures \ref{fig:best_response_intensity} and \ref{fig:cA_cB_phase_diagram}. The stripes in Figure \ref{fig:best_response_proportional} denote regions where $\IS$ is de-monopolized.} 
\label{fig:proportional_rationing_plots}
\end{figure}

\subsection{Imperfect substitutability} 
\label{sec:imperfect_substitutability}

For the sake of simplicity and to capture main ideas without excessive complexity, we have thus far assumed that the two players sell perfectly substitutable products. However, in practice, sellers often sell products that are \emph{imperfectly substitutable}, meaning that a consumer may have an innate preference for one seller's product over the other's and thus does not always necessarily buy the cheaper product. 
Unfortunately, the Anna Karenina principle applies to substitutability: \emph{Perfectly substitutable goods are all alike; every imperfectly substitutable good is imperfectly substitutable in its own way.} 
More explicitly, with perfect substitutability, player $i$'s demand function $D_i$ can be expressed in terms of the original demand function $Q$ and residual demand function $R$, and there is a discontinuous jump from $Q$ to $R$ when player $i$'s price $p_i$ goes above player $j$'s price. However, under imperfect substitutability, $D_i$ may more generally depend continuously on $p_i - p_j$.\footnote{Traditional models of markets of differentiated products (see Chapter 6 of \citet{vives1999oligopoly}) generally start with defining the utility function $U(q_1, q_2)$ of a representative consumer and then maximizing the net utility: $\max_{q_1, q_2} U(q_1, q_2) - p_1 q_1 - p_2 q_2$. The first-order conditions yield the inverse demands $p_i = \pd{q_i}U(q_1, q_2)$ for $i=1,2$. Solving this system for $q_i$'s yields the direct demands. This approach is nice because it derives demand curves from first principles. Deriving equilibria for such models is out of scope for this paper but is a good direction for future work.} Since there is an infinite number of forms this dependence could take, it is difficult to derive general equilibrium results under imperfect substitutability. 
However, it is straightforward to show that, under any reasonable model of imperfect substitutability, it is always possible for the marketplace operator to de-monopolize the independent seller. In the imperfect substitutability setting, we redefine $\IS$'s sole-seller price to be $\psole = \arg \max_{p_{\IS}} D_{\IS}(p_{\IS}; 0, \infty) p_{\IS}$.

\begin{proposition} [Feasibility of de-monopolization under imperfect substitutability]
\label{prop:imperfect_substitutability_general_result}
\begin{samepage}
    Suppose $D_{\IS}(p_{\IS}; q_{\MO}, p_{\MO})$ satisfies the following conditions.
    \begin{enumerate}
        \item (Positive substitutability) There exists $\tilde p_{\MO}$ and $\tilde q_{\MO}$ such that
        $D_{\IS}(p_{\IS}; 0, \infty) > D_{\IS}(p_{\IS}; \tilde q_{\MO}, \tilde p_{\MO})$ for all $p_{\IS}$.
        \item (Nondecreasing gap) For fixed marketplace operator actions $(p_1, q_1)$ and $(p_2, q_2)$ with $p_1 > p_2$ and $q_1 < q_2$, the gap between the independent seller's demand functions is constant or increasing with $p_{\IS}$ — that is, $$\frac{d}{dp_{\IS}} \Big(D(p_{\IS}; q_1, p_1) - D(p_{\IS}; q_2, p_2)\Big) \geq 0$$
    \end{enumerate}
    Then, the independent seller's best response to  $(\tilde p_{\MO}, \tilde q_{\MO})$ will be to set a price strictly lower than their sole-seller price $\psole$. 
\end{samepage}
\end{proposition}

The proof is given in Appendix \ref{sec:auxiliary_proofs_APPENDIX}.

\paragraph{A one-directional substitutability model.}
Although it is hard to derive equilibrium results that hold for every substitutability model, it is easy to extend our existing results to a one-directional substitutability model, which we now describe.
Let $\gamma \in [0,1]$ be a substitutability parameter, where $\gamma=1$ recovers perfect substitutes and $\gamma=0$ means that the two goods are not at all substitutable.
For each of the two rationing rules we have considered, we can obtain the imperfect substitutes version of the residual demand function for player $i$ by replacing $q_j$ with $\gamma q_j$. For intensity rationing, this yields $R(p_i; q_j) = \theta - p_i - \gamma q_j$. For proportional rationing, this yields $ R(p_i; q_j, p_j) = (\theta - p_i)\left(1-\frac{\gamma q_j}{Q(p_j)}\right)$.

\emph{Interpretation.}
Under this one-directional substitutability model, the effect of the inventory of the lower-priced seller $i$ on the demand for the higher-priced seller $j$ is now damped by a $\gamma$ factor. The demand for the lower-priced seller is not affected by $\gamma$; only the residual demand for the higher-priced seller is affected (hence ``one-directional'' imperfect substitutability). The interpretation is something like this: Each consumer initially wants both player $i$'s good and player $j$'s good and assigns the same valuation $v$ to both (this valuation is randomly sampled from $\mathrm{Unif}([0,\theta])$. If they buy the good from player $i$, then with probability $\gamma$ they will no longer want the good from player $j$. Otherwise, they keep their original valuation, and they buy from player $j$ if $v \geq p_j$. Which customers are able to buy at the lower price is determined by the rationing rule.

\begin{figure}[h]
\centering
\begin{subfigure}{.45\textwidth}
  \centering
  \includegraphics[width=.7\textwidth]{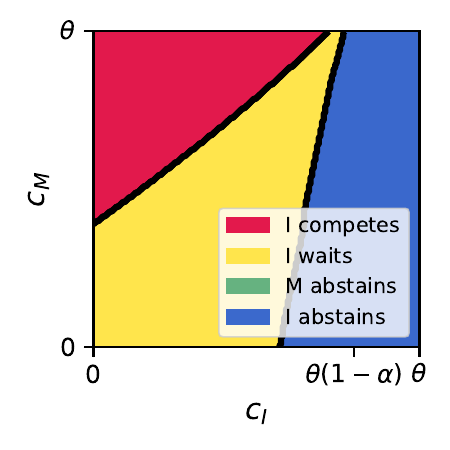}
  \vspace{-10pt}  
  \caption{Intensity}
\end{subfigure}
\begin{subfigure}{.45\textwidth}
  \centering
\includegraphics[width=.7\textwidth]{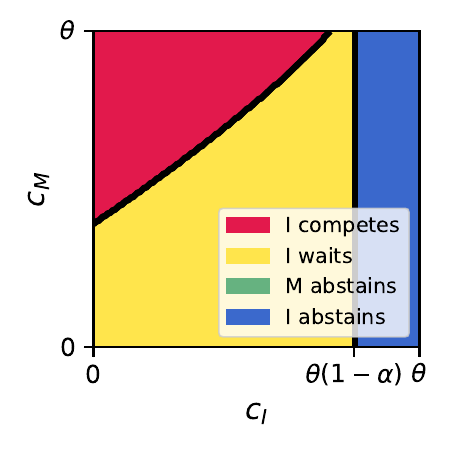}
    \vspace{-10pt}  
  \caption{Proportional}
\end{subfigure}
\caption{Equilibrium strategies under one-directional imperfect substitutability with $\gamma=0.5$
for different $c_{\MO}, c_{\IS}$ combinations in the same setting as Figures \ref{fig:cA_cB_phase_diagram} and \ref{fig:cA_cB_phase_diagram_proportional}.}
\label{fig:gamma=0.5_cAcB_phase_diagram}
\end{figure}

Applying results from Appendix \ref{sec:proofs_APPENDIX} allows us to solve for the equilibrium under one-directional imperfect substitutability. 
We visualize the equilibrium strategies for $\gamma=0.5$ in Figure \ref{fig:gamma=0.5_cAcB_phase_diagram}. Compared to Figures \ref{fig:cA_cB_phase_diagram} and \ref{fig:cA_cB_phase_diagram_proportional}, which are the same setting but with $\gamma=1$ (perfect substitutability), we see a larger yellow region (corresponding to $\IS$ waiting it out) in the imperfect substitutes setting.

\section{Conclusion}

Monopolist sellers have free rein when setting prices, which often leads to unhappy customers. Our paper shows that, when faced with such sellers, marketplace operators can apply competitive pressure by acting as a credible low-price competitor. This credibility is achieved by stocking sufficient inventory. We find that, even if this inventory is not sold, this operator-as-seller configuration is often utility-maximizing for the marketplace operator and is always beneficial to consumers. Since more inventory is being offered at lower prices, consumer surplus is necessarily increased.
Our analysis is robust to perturbations of our model (e.g., assumptions about demand rationing and substitutability). To summarize, market entry is a useful tool that marketplace operators can consider if they believe regulation to be too heavy-handed.

\paragraph{Future directions.}
We view our work as a first step towards understanding the operator-as-seller market structure. However, there is still much to be done. We focus on a simple model for the purpose of capturing the main ideas without being burdened by technical complexities, but there are multiple ways in which our model can be made more realistic or capture additional real-world scenarios.  We now list several. \emph{(1) Multiple third-party sellers.} Our current analysis relies on the assumption that once a seller runs out of inventory, all remaining customers can only buy from the remaining seller. Having multiple third-party sellers may give rise to new inter-seller dynamics. \emph{(2) Replacing $k$ with a multiplier of consumer surplus.} Recall that $k$ is meant to model the marketplace operator's extra benefit (beyond the immediate monetary benefit) that arises from a customer's satisfaction for having made a purchase. Our current model represents this as a flat per-unit utility. However, it is arguably more realistic to model $k$ as dependent on price or as a multiplier of consumer surplus, since a customer who buys a product at \$5 is likely happier than a customer who buys a product at \$10 and is more likely to return to the platform and make more purchases. \emph{(3) Positive salvage value.} This work implicitly assumes that a player's leftover inventory has zero salvage value. In reality, products that have steady demand generally have positive salvage value because excess inventory can be sold in the next time period. 
\emph{(4) Non-linear demand.} 
We assume a linear demand in our paper, and it would be useful to verify whether our findings hold under other demand functions (e.g., Cobb-Douglas) as a robustness check.
\emph{(5) Integer-constrained inventory.}
Our model assumes quantity is continuous and that fractional amounts of inventory can be ordered. In reality, inventory quantities are often integers. It is worth investigating whether anything unexpected happens if quantities are restricted to be integer-valued.
\emph{(6) Endogenous timing.} Suppose a platform can choose between the Stackelberg game that we analyze in this paper and a simultaneous game (where $\MO$ and $\IS$ make their decisions at the same time). Which game yields higher utility to $\MO$?  \emph{(7) Welfare under other rationing rules.} Investigating the welfare implications under rationing rules other than intensity rationing with perfect substitutability is important from a practical perspective, as many products have demand rationed in other ways. 

\newpage

\setlength{\bibsep}{4pt} 
\bibliographystyle{plainnat}
\bibliography{references}

\newpage
\appendix


\section{Proofs} 
\label{sec:proofs_APPENDIX}

In the main paper, we stated results for intensity rationing (\Cref{lemma:best_response_intensity}, \Cref{lemma:qstar_A}, \Cref{theorem:equilibrium_unconstrained_game}), proportional rationing (\Cref{sec:proportional_rationing}), and imperfect substitutability (\Cref{sec:imperfect_substitutability}) separately for the sake of presentation. However, these results all follow from the results we present in this section that characterize $\IS$'s best response and game equilibrium for general classes of residual demand functions. Observe that intensity rationing and proportional rationing are special cases of the following classes of residual demand functions.

\begin{enumerate}
    \item \textbf{Intensity rationing} (also known as efficient rationing): Assuming that the highest-valuation customers buy first (i.e., at the lower price) yields the residual demand function
    $$R(p_i) = Q(p_i) - q_j.$$
    This is a member of the class of \emph{subtractive residual demand functions}, which are functions that can be represented as
    $$R(p_i) = Q(p_i) - f(q_j)$$
    for some $f: [0, \theta] \to [0, \theta]$ that is an increasing function of $q_j$ satisfying $f(0) = 0$ and $f(q) \leq q$ for all $q \in [0, \theta]$.
    \item \textbf{Proportional rationing} (also known as Beckmann rationing; see \citet{davidson1986long}): Suppose that every customer has the same probability $\rho$ of being assigned to the lower-priced seller, where $\rho$ is such that in expectation the demand that the lower-price seller faces at their price $p_i$ is exactly equal to their supply $q_i$. This yields
    $$R(p_i) = Q(p_i)  \left(1-\frac{q_j}{Q(p_j)}\right).$$
    This is a member of the class of \emph{multiplicative demand functions}, which are functions that can be represented as
    $$R(p_i) = Q(p_i) g(q_j)$$
    where $g: [0, \theta] \to [0, 1]$ is a decreasing function of $q_j$ satisfying $g(0) = 1$. $g$ can also depend on $p_j$.\footnote{Note that, technically speaking, all of the expressions for residual demand functions should include a $\min(\mathrm{RHS}, 0)$ to ensure the non-negativity in Assumption \ref{assumption:R}(b) is satisfied. However, the assumptions we place on $f$ and $g$ ensure that $R(p_i)$ for subtractive or multiplicative residual demand functions satisfies non-negativity for any $q_j \in [0, Q(p_j)]$. In other words, $\min(\mathrm{RHS}, 0) = \mathrm{RHS}$ for any $q_j \in [0, Q(p_j)]$. We thus omit the $\min(\mathrm{RHS}, 0)$ since any utility-maximizing player $j$ will never set $q_j > Q(p_j)$.}
\end{enumerate}

Note that the one-directional imperfectly substitutable variants of intensity and proportional rationing described in \Cref{sec:imperfect_substitutability} are members of the additive and multiplicative demand function classes, respectively. Additive and multiple demand functions are just two examples of residual demand function classes. More generally, residual demand functions should satisfy the following conditions.

\begin{assumption} \label{assumption:R}
The residual demand function $R(p_i; q_j, p_j)$, where $p_i$ denotes the price of the higher-priced seller and $q_j$ and $p_j$ are the inventory and price of the lower-priced seller, satisfies
\begin{enumerate}[(a)]
    \item $R(p_i; q_j, p_j) \leq Q(p_i)$ for all $p_i\in[0,\theta]$ --- the residual demand cannot exceed the original demand at the same price. 
    \item $R(p_i; q_j, p_j) \geq 0$ for all $p_i\in[0,\theta]$ --- the residual demand cannot be negative.
    \item  $\pd{p_i}R(p_i;q_j, p_j) \leq 0$ --- as the higher-priced seller increases their price, their demand decreases or does not change. 
    \item $\pd{q_j}R(p_i;q_j, p_j) \leq 0$ --- if the lower-priced seller increases their inventory, this decreases or does not change the demand left over for the higher-priced seller.
\end{enumerate}
To avoid clutter we will sometimes omit the dependence of $R$ on $q_j$ and $p_j$ and simply write $R(p_i)$.
\end{assumption}


\subsection{Best response proofs} \label{sec:best_response_proofs_APPENDIX}

We will now characterize $\IS$'s best response when facing an additive or multiplicative residual demand function. We do so via three propositions. The first of these propositions describes the independent seller's optimal strategy in the straightforward case where $\MO$ has set a ``noncompetitive'' (too large) price. The result is intuitive: $\psole$ is the price that $\IS$ would set if the marketplace operator were not a seller, and when the operator does not apply any competitive pressure by setting a lower price, there is no reason for $\IS$ to not set their price to be $\psole$.

\begin{proposition} [Large $p_{\MO}$ best response]
    \label{prop:best_response_large_pA}
    Whenever $p_{\MO} \geq \psole$ ($\MO$'s price is higher than $\IS$'s optimal sole-seller price), $\IS$ passively competes by setting $p_{\IS} = \psole$.
\end{proposition}

\begin{proof}[Proof of Proposition \ref{prop:best_response_large_pA}]

    We first argue that the optimal $p_{\IS} \leq p_{\MO}$ is $p_{\IS} = \psole$.
    Since we assumed $\psole \leq p_{\MO}$, we have $\psole = \arg \max_p ((1-\alpha)p - c_{\IS}) Q(p) = \arg \max_{p \leq p_{\MO}} ((1-\alpha)p - c_{\IS}) Q(p)$, where the first equality is the definition of $\psole$ and the second equality comes from the fact that the maximizer of a set is the maximizer of any subset containing the maximizer. 

     Let $u_{\IS}(p)$ denote $\IS$'s utility when they set price $p$ and buy inventory to exactly meet their demand at that price. We now show that $u_{\IS}(\psole) > u_{\IS}(p_{\IS})$ whenever $p_{\IS} \geq p_{\MO}$. Consider any $p_{\IS} \geq p_{\MO}$. Then,
    \begin{align*}
        u_{\IS}(\psole) &= [(1-\alpha)\psole - c_{\IS}]Q(\psole) \\
        &> [(1-\alpha)p_{\IS} - c_{\IS}]Q(p_{\IS}) \qquad \text{since $\psole$ is the unique maximizer} \\
        &\geq [(1-\alpha)p_{\IS} - c_{\IS}]R(p_{\IS}) \qquad \text{Assumption \ref{assumption:R}(a)} \\
        &= u_{\IS}(p_{\IS}).
    \end{align*}


    We have shown that $\psole$ is the optimal $p_{\IS}$ among $p_{\IS} \leq p_{\MO}$ and it achieves higher utility than any $p_{\IS} \geq p_{\MO}$, so it is the global maximizer of utility. 
\end{proof}

In the next proposition, we will consider the more interesting case where $p_{\MO}$ is at an intermediate level ($p_{\MO} \in [p_0, \psole)$), and $\MO$'s inventory level $q_{\MO}$ determines what strategy $\IS$ will take. In other words, $\IS$ must determine if it is worth it to wait it out or if, by the time $\MO$ sells out, there will be so little demand left that it would be better for $\IS$ to compete.
This is because whenever $p_{\MO} \geq p_0$, the independent seller will never abstain because positive utility is achievable (or in the case of $p_{\MO}=p_0$, non-negative utility is achievable, but as described before we assume that $\IS$ is willing to sell at price $p_0$). Thus, we simply have to determine whether $\IS$ competes or waits. The following proposition describes the critical $q_{\MO}$ threshold that determines which $\IS$ strategy is optimal. This
is obtained by first determining the optimal price that $\IS$ should set when they compete ($p_{\MO}$) and when they wait ($p_{\IS}^W$), then comparing the corresponding utilities and choosing the price that achieves the higher utility. The critical threshold $\qthresh(p_{\MO})$ is simply the value of $q_{\MO}$ below which competing yields higher utility and above which waiting yields higher utility. For a function $v: \R \to \R$, let $v^{-1}:\R \to \R$ denote its inverse function.\footnote{The inverse $v^{-1}$ is well-defined when $v$ is one-to-one (e.g., strictly increasing or strictly decreasing). We additionally extend the notion of inverse to functions that are increasing or decreasing (but not strictly) in a way that is useful for our analysis. 
For an increasing function $f$, we let $f^{-1}(y) = \inf\{x: f(x) \geq y\}$, and for a decreasing function $g$, we let  $g^{-1}(y) = \inf\{x: g(x) \leq y\}$.}

\begin{proposition} [Intermediate $p_{\MO}$ best response]
    \label{prop:best_response_intermediate_pA}
    Let $p_{\MO} \in [p_0, \psole)$. The independent seller's response depends on $\MO$'s inventory, relative to a threshold $\qthresh(p_{\MO})$:
    if $q_{\MO} \geq \qthresh(p_{\MO})$, the independent seller competes by setting $p_{\IS}=p_{\MO}$; otherwise they wait it out by setting $p_{\IS}=p^W_{\IS}$. The threshold function $\qthresh(p_{\MO})$ and wait it out price $p^W_{\IS}$ take different forms depending on the residual demand function:
    \begin{enumerate}[(a)]
        \item For a subtractive residual demand function, 
        $$\qthresh(p_{\MO}) = f^{-1}\left(\theta - p_0 - 2 \sqrt{(p_{\MO} - p_0)(\theta-p_{\MO})}\right) \qquad \text{and} \qquad p_{\IS}^W := \psole - \frac{f(q_{\MO})}{2}.$$
        \item For a multiplicative residual demand  function,
        $$\qthresh(p_{\MO}) = g^{-1}\left(\frac{(p_{\MO} - p_0)(\theta-p_{\MO})}{(\psole-p_0)(\theta-\psole)}\right)\qquad \text{and} \qquad p_{\IS}^W := \psole.$$
    \end{enumerate}
\end{proposition}

\begin{proof}[Proof of Proposition \ref{prop:best_response_intermediate_pA}] 
    Recall that $\IS$ has three strategies to choose from: abstain, compete, or wait.
    Whenever $p_{\MO} \geq p_0$, $\IS$ can achieve non-negative utility by competing, so they will never abstain.
    Thus, to determine what $\IS$ will do, we simply have to determine whether it is better for them to compete or wait.
    \begin{enumerate}
        \item \emph{Optimal ``compete'' utility.} In order to compete, $\IS$ must set $p_{\IS} \leq p_{\MO}$. The optimal $p_{\IS}$ satisfying this condition is
        \begin{align*}
            p_{\IS}^{C} &:= \arg \max_{p_{\IS} \leq p_{\MO}} u_{\IS}(p_{\IS}; p_{\MO}, q_{\MO}) \\
            &= \arg \max_{p_{\IS} \leq p_{\MO}} \underbrace{((1-\alpha)p_{\IS} - c_{\IS}) (\theta - p_{\IS})}_{:=u^C(p_{\IS})} \\
            &= p_{\MO}
        \end{align*}
        where the last equality comes from combining the following facts: (1) $\psole$ is the unconstrained maximizer of $u^C(p_{\IS})$; (2) the conditions of the proposition assume $p_{\MO} \in [p_0, \psole)$; and (3) $u^C(p_{\IS})$ is convex so the $p_{\IS}$ in the constraint set that is closest to the unconstrained maximizer must be constrained maximizer. 
        When $\IS$ sets $p_{\MO}$, they get utility 
        $$u_C^* := u^C(p_{\IS}^C) = ((1-\alpha)p_{\MO} - c_{\IS}) (\theta - p_{\MO}).$$
        Note that the optimal compete utility does not depend on the residual demand function because when $\IS$ competes, they do not face the residual demand function.
        \item \emph{Optimal utility when facing $R$ (upper bound on optimal ``wait'' utility).} To derive the optimal ``wait'' utility, we would have to solve 
        $$\max_{p_{\IS} > p_{\MO}} ((1-\alpha)p_{\MO} - c_{\IS}) R(p_{\IS}; q_{\MO}, p_{\MO}).$$
        Rather than solving this constrained optimization problem, we solve the unconstrained optimization problem (for reasons to be explained shortly):
        $$\max_{p_{\IS}} ((1-\alpha)p_{\MO} - c_{\IS}) R(p_{\IS}; q_{\MO}, p_{\MO}).$$
        Lemma \ref{lemma:optimal_pB_when_facing_R} gives the $p_{\IS}$ that solves this for each class of residual demand functions.
        Call this $p_{\IS}^R$ and let
        $$u^*_R := ((1-\alpha)p_{\IS}^R - c_{\IS}) R(p_{\IS}^R; q_{\MO}, p_{\MO})$$
        denote the maximal value achieved by the unconstrained optimization problem.
    \end{enumerate}
        We now make the following observation:
        \begin{center}
            \emph{$\IS$ gets higher utility from waiting than competing if and only if $u^*_R > u^*_C$.}
        \end{center}
        This is stated formally and proven in Lemma \ref{lemma:uR_greater_than_uC_equivalency}.
        Thus, to determine whether it is optimal for $\IS$ to wait by setting $p_{\IS}^R$ or compete by setting $p_{\IS}^C$, we simply have to check if $u^*_R > u^*_C$. For each class of residual demand functions,
        we will solve this condition for $q_{\MO}$  by plugging in the optimal $p_{\IS}$ values derived in Lemma \ref{lemma:optimal_pB_when_facing_R}.
        \begin{enumerate}
            \item \emph{Subtractive demand functions.} 
            First, observe that $u_R^*$ for subtractive demand functions can be written as 
            \begin{align*}
                u_R^* &= \left((1-\alpha)\left(\psole - \frac{f(q_{\MO})}{2}\right)-c_{\IS}\right)\left(\theta - \left(\psole - \frac{f(q_{\MO})}{2}\right) -f(q_{\MO})\right) \\
                &= (1-\alpha) \left(\psole - \frac{f(q_{\MO})}{2} - p_0\right)\left(\theta - \psole - \frac{f(q_{\MO})}{2} \right) \qquad \text{since $p_0 = \frac{c_{\IS}}{1-\alpha}$} \\
                &= (1-\alpha) \left(\frac{p_0 + \theta}{2} - \frac{f(q_{\MO})}{2} - p_0\right)\left(\theta - \frac{p_0 + \theta}{2} - \frac{f(q_{\MO})}{2}\right) \qquad \text{definition of $\psole$} \\
                &= (1-\alpha) \left(\frac{\theta}{2} - \frac{p_0}{2} - \frac{f(q_{\MO})}{2}\right)\left(\frac{\theta}{2} - \frac{p_0}{2} - \frac{f(q_{\MO})}{2}\right) \\
                &= \frac{(1-\alpha)}{4}(\theta - p_0 - f(q_{\MO}))^2.
            \end{align*}
            Thus,
            \begin{align*}  u_R^* > u_C^* \Longleftrightarrow \frac{(1-\alpha)}{4}(\theta - p_0 - f(q_{\MO}))^2 &> ((1-\alpha)p_{\MO} - c_{\IS})(\theta - p_{\MO}).
            \end{align*}
            Solving for $q_{\MO}$, we get
            \begin{align*}
            u_R^* > u_C^* &\Longleftrightarrow
                f(q_{\MO}) < \theta - p_0 - 2 \sqrt{(p_{\MO} - p_0)(\theta - p_{\MO})}.\\
                &\Longleftrightarrow q_{\MO} < f^{-1}\big(\theta - p_0 - 2 \sqrt{(p_{\MO} - p_0)(\theta - p_{\MO})}\big).
            \end{align*}
            \item \emph{Multiplicative demand functions.} 
            \begin{align*}  
             u_R^* > u_C^* &\Longleftrightarrow
             ((1-\alpha)\psole - c_{\IS})(\theta - \psole)g(q_{\MO}) > ((1-\alpha)p_{\MO} - c_{\IS})(\theta - p_{\MO}) \\
             &\Longleftrightarrow
             (1-\alpha)(\psole - p_0)(\theta - \psole)g(q_{\MO}) > (1-\alpha)(p_{\MO} - p_0)(\theta - p_{\MO}) \quad \text{since $p_0 = \frac{c_{\IS}}{1-\alpha}$} \\
              &\Longleftrightarrow
             (\psole - p_0)(\theta - \psole)g(q_{\MO}) > (p_{\MO} - p_0)(\theta - p_{\MO}) \\
             &\Longleftrightarrow g(q_{\MO}) > \frac{(p_{\MO} - p_0)(\theta - p_{\MO})}{(\psole - p_0)(\theta - \psole)} \\
             &\Longleftrightarrow q_{\MO} < g^{-1}\left(\frac{(p_{\MO} - p_0)(\theta - p_{\MO})}{(\psole - p_0)(\theta - \psole)}\right).
             \end{align*}
           
        \end{enumerate}
    
\end{proof}

\begin{lemma} [Optimal $p_{\IS}$ when facing each $R$]
\label{lemma:optimal_pB_when_facing_R}
The solution of 
\begin{align} 
    \arg \max_{p_{\IS}} u_R(p_{\IS})
\end{align}
where $u_R(p_{\IS}) = ((1-\alpha)p_{\MO} - c_{\IS}) R(p_{\IS})$
is...
\begin{enumerate}[(a)]
    \item For subtractive residual demand function $R(p_{\IS})=\theta -p_{\IS} - f(q_{\MO})$, 
        $$p_{\IS}^{\mathrm{sub}} = \psole - \frac{f(q_{\MO})}{2}.$$
    \item For multiplicative residual demand function $R(p_{\IS}) = (\theta-p_{\IS})g(q_{\MO})$,
    $$p_{\IS}^{\mathrm{mult}} = \psole.$$
\end{enumerate}
\end{lemma}
\begin{proof} Since all of the residual demand functions we consider result in concave $u_R(p_{\IS})$, the vertex of $u_R(p_{\IS})$ is the maximizer.
    \begin{enumerate}[(a)]
        \item Plugging the subtractive residual demand function into $u_R(p_{\IS})$ yields
        $$u_R(p_{\IS}) = ((1-\alpha)p_{\MO} - c_{\IS}) (\theta - p_{\MO} - f(q_{\MO})).$$
        This has roots at $p_{\MO} = \frac{c_{\IS}}{1-\alpha} = p_0$ and $p_{\MO} = \theta - f(q_{\MO})$, so the vertex is at $\frac{p_0 + \theta - f(q_{\MO})}{2} = \psole - \frac{f(q_{\MO})}{2}$.
        \item Plugging the multiplicative residual demand function into $u_R(p_{\IS})$ yields
        $$u_R(p_{\IS}) = ((1-\alpha)p_{\MO} - c_{\IS}) (\theta - p_{\MO})h(q_{\MO}),$$
        which is the same utility function that $\IS$ faces if they are the sole seller, but rescaled. Since the rescaling does not change the location of the maximizer, the maximizer of this $u_R(p_{\IS})$ is $\psole$.
    \end{enumerate}
\end{proof}

\begin{lemma}
    \label{lemma:uR_greater_than_uC_equivalency}
    Define
    $u_C^* = \max_{p_{\IS} \leq p_{\MO}} ((1-\alpha)p_{\MO} - c_{\IS}) Q(p_{\IS})$
    to be $\IS$'s optimal utility for competing,
    $u_W^* = \max_{p_{\IS} > p_{\MO}} ((1-\alpha)p_{\MO} - c_{\IS}) R(p_{\IS}; q_{\MO}, p_{\MO})$
    to be $\IS$'s optimal utility for waiting,
    and $u_R^* = \max ((1-\alpha)p_{\MO} - c_{\IS}) R(p_{\IS}; q_{\MO}, p_{\MO})$
    to be $\IS$'s optimal utility when facing the residual demand function.
    Then,
    \begin{center}
        $u_W^* > u_C^*$ if and only if $u_R^* > u_C^*$.
    \end{center}
\end{lemma}
\begin{proof} We prove both directions:
\begin{itemize} 
    \item ($\Longrightarrow$) Suppose $u_W^* > u_C^*$. Since $u_R^* \geq u_W^*$, we have $u_R^* \geq u_W^* > u_C$. 
    \item ($\Longleftarrow$) Suppose $u_R^* > u_C^*$. By Assumption \ref{assumption:R}(a), which says that $R(p_{\IS})$ is no higher than $Q(p_{\IS})$ for all $p_{\IS}$, we know that the $p_{\IS}$ corresponding to $u_R^*$ must be at least $p_{\MO}$. So we must have $u_W^*=u_R^* > u_C^*$.
\end{itemize}
    
\end{proof}

Finally, we consider the case where $\MO$ has set a very low price ($p_{\MO} \leq p_0$). Much like Proposition \ref{prop:best_response_intermediate_pA}, the following proposition determines the independent seller's best response by comparing $q_{\MO}$ to a critical threshold. The difference in this case is that the strategies $\IS$ is choosing between are to wait or to abstain; competing is no longer a viable option because, by definition of $p_0$, $\IS$ is not able to make a profit if they set a price below $p_{\MO} < p_0$. 

\begin{proposition} [Small $p_{\MO}$ best response]
    \label{prop:best_response_small_pA}
    If $p_{\MO} < p_0$, the independent seller's response depends on $\MO$'s inventory, relative to a threshold $\qddag(p_{\MO})$:
    the independent seller abstains by setting $p_{\IS}=\infty$ if $q_{\MO} \geq \qddag(p_{\MO})$ and otherwise sets the wait it out price $p_{\IS}^W$ described in Proposition \ref{prop:best_response_intermediate_pA}. The value of the threshold $\qddag$ depends on the residual demand function:
    \begin{enumerate}[(a)]
        \item For a subtractive residual demand function,
        $$\qddag(p_{\MO})=f^{-1}(\theta - p_0)$$
        \item For a multiplicative residual demand function,  $$\qddag(p_{\MO})=g^{-1}(0; p_{\MO}),$$
        where we write $g^{-1}(~\cdot~; p_{\MO})$ to emphasize that $g$ can depend on $p_{\MO}$.    
    \end{enumerate}
\end{proposition}

\begin{proof}[Proof of Proposition \ref{prop:best_response_small_pA}]
Since $p_{\MO} < p_0$, competing yields negative utility for $\IS$, so they must decide between waiting or abstaining. Abstaining yields zero utility, so $\IS$ will abstain only if no positive utility can be achieved by waiting. This happens if there are no customers left that are willing to pay a price higher than $\IS$'s break-even price $p_0$.
$\qddag(p_{\MO})$ is the smallest $q_{\MO}$ such that after $\MO$ sells $q_{\MO}$ units at price $p_{\MO}$, there is zero residual demand at $\IS$'s break-even price, i.e., $\qddag(p_{\MO})=\min \{q_{\MO}: R(p_0;p_{\MO}, q_{\MO}) = 0\}$. 
    For any $q_{\MO} < \qddag(p_{\MO})$, we have $R(p_0;p_{\MO}, q_{\MO}) > 0$.
    
    If $R(p_0;p_{\MO}, q_{\MO}) = 0$, this means that there is no price at or above $\IS$'s break-even price at which there is positive demand (this is true because $R$ is non-increasing in its first argument by Assumption \ref{assumption:R}). In these cases, the best $\IS$ can do is get zero utility by abstaining. When $R(p_0;p_{\MO}, q_{\MO}) > 0$, $\IS$ should set their optimal price when facing $R$, which is given exactly by $p_{\IS}^W$ described in Proposition \ref{prop:best_response_intermediate_pA}. 
\end{proof}

\subsection{Equilibrium proofs} 
\label{sec:equilibrium_APPENDIX}

We prove the main result that characterizes $\MO$'s equilibrium strategy.

\begin{proof}[Proof of Lemma \ref{lemma:qstar_A}] We split into three cases.

\emph{Case 1: $p_{\MO} \geq \psole$.} By Proposition \ref{prop:best_response_large_pA}, we know that whenever $p_{\MO} \geq \psole$, $\IS$ will compete regardless of the value of $q_{\MO}$. Thus, the only implementable strategy is to induce $\IS$  to compete. The optimal $q_{\MO}$ that induces $\IS$  to compete is $q_{\MO}=0$ since $\MO$'s inventory is not sold and has no effect on $\IS$'s actions.

\emph{Case 2: $p_{\MO} \in [p_0, \psole)$.} 
By Proposition \ref{prop:best_response_intermediate_pA}, we know that when $p_{\MO} \in [p_0, \psole)$, $\IS$  will either compete or wait depending on the value of $q_{\MO}$. The optimal $q_{\MO}$ that induces $\IS$ to compete is $\qthresh(p_{\MO})$ because this is the smallest $q_{\MO}$ value that induces compete, and none of $\MO$'s inventory gets to be sold when $\IS$ competes. 
The range of $q_{\MO}$ values that induce $\IS$ to wait is $\cW := [0, \qthresh(p_{\MO}))$.
We now argue that, under the assumptions on $R$ imposed by the lemma, the optimal $q_{\MO}$ that induces $\IS$  to wait must be at the boundaries. We will do this by showing that the second derivative is non-negative for all such $q_{\MO}$, implying there are no local maximizers.
    
For any $q_{\MO}$ that induces $\IS$ to wait, $\MO$'s utility function is 
    \begin{align} \label{eq:uA_for_induce_wait}
        u_{\MO}(p_{\MO}, q_{\MO}) = (p_{\MO} - c_{\MO} + k) q_{\MO} + (\alpha \pBR(p_{\MO}, q_{\MO}) + k) R(\pBR(p_{\MO}, q_{\MO}); q_{\MO}, p_{\MO}).
    \end{align}
    \begin{enumerate}
        \item For subtractive residual demand functions $R(p_{\IS}; q_{\MO}) = \theta - p_{\IS} -f(q_{\MO})$, we have $\pBR(p_{\MO}, q_{\MO}) = \psole - \frac{f(q_{\MO})}{2}$, so \eqref{eq:uA_for_induce_wait} becomes 
        $$u_{\MO}(p_{\MO}, q_{\MO}) = (p_{\MO} - c_{\MO} + k) q_{\MO} + \left(\alpha\left(\psole - \frac{f(q_{\MO})}{2}\right)+k\right)\left(\theta - \psole - \frac{f(q_{\MO})}{2}\right),$$
        which has second derivative $\frac{\partial^2}{\partial q_{\MO}^2} u_{\MO} = \frac{\alpha}{2}\left(f'(q_{\MO})\right)^2 > 0$  when $f''(q_{\MO}) = 0$.
        \item For multiplicative demand functions $R(p_{\IS}; q_{\MO}) = (\theta - p_{\IS})g(q_{\MO})$, we have $\pBR(p_{\MO}, q_{\MO}) = \psole$, so \eqref{eq:uA_for_induce_wait} becomes
        $$u_{\MO}(p_{\MO}, q_{\MO}) = (p_{\MO} - c_{\MO} + k) q_{\MO} + (\alpha \psole + k) (\theta - \psole) g(q_{\MO}),$$
        which has second derivative $\frac{\partial^2}{\partial q_{\MO}^2} u_{\MO}=0$ under the assumption that $g''(q_{\MO}) = 0$.
    \end{enumerate}
    Thus, the candidates for the optimal $q_{\MO}$ that induces $\IS$ to wait are the boundary points of $\cW$ --- namely, zero and $\qthresheps(p_{\MO})$. Combining with the optimal $q_{\MO}$ that induces $\IS$ to wait tells us that the  overall optimal $q_{\MO}$ must come from the set $\{0, \qthresheps(p_{\MO}), \qthresh(p_{\MO})\}$.

\emph{Case 3: $p_{\MO} < p_0$.} By Proposition \ref{prop:best_response_small_pA}, we know that when $p_{\MO} < p_0$, $\IS$ will either wait or abstain. To induce wait when $p_{\MO} < p_0$, $\MO$ must set $q_{\MO} \in [0, \qddag(p_{\MO}))$. By the same second derivative argument as in Case 2, the maximizing quantity must occur at the boundary (0 or $\qddageps(p_{\MO})$) where $\qddageps(p_{\MO}) = \qddag(p_{\MO}) -\epsilon$, for small $\epsilon$, is the largest $q_{\MO}$ such that $\IS$ waits.
To induce abstain, $\MO$ must set $q_{\MO} \geq \qddag(p_{\MO})$ and they will not set $q_{\MO} \geq Q(p_{\MO})$ because inventory exceeding demand is not sold. When $\IS$ waits, $\MO$'s utility is $u_{\MO} = (p_{\MO} - c_{\MO} + k) q_{\MO}$. If $p_{\MO} - c_{\MO} + k < 0$, then the highest utility is achieved at the boundary point $q_{\MO} = \qddag(p_{\MO})$, which yields a negative utility. If $p_{\MO} - c_{\MO} + k > 0$, then the highest utility is achieved at the boundary point $q_{\MO} = Q(p_{\MO})$.
Thus, the candidates for the optimal $q_{\MO}$ are zero and $\qddageps(p_{\MO})$ (which induce wait) and $\qddag(p_{\MO})$ and $Q(p_{\MO})$ (which induce abstain). It is easy to see that $u_{\MO}(p_{\MO}, \cdot)$ is continuous at $\qddag(p_{\MO})$, so $u_{\MO}(p_{\MO}, \qddag(p_{\MO})) = u_{\MO}(p_{\MO}, \qddageps(p_{\MO}))$. Thus, we only have to compare $\MO$'s utilities at $q_{\MO} \in \{0,\qddag(p_{\MO}), Q(p_{\MO})\}$ However, as previously argued, if $\qddag(p_{\MO})$ achieves higher utility than $Q(p_{\MO})$, then that utility is negative, which would be less than the utility of setting $q_{\MO} = 0$. As a result, it is sufficient to just compare the utilities at $q_{\MO} \in \{0, Q(p_{\MO})\}$. 

\end{proof}

\noindent \emph{Elaboration on Remark \ref{remark:efficient_computation_of_equilibrium}.} 
We now explain how the equilibrium can be computed efficiently. Figure \ref{fig:equilibrium_opt} illustrates how we identify $p^*_{\MO}$. 
From Lemma \ref{lemma:qstar_A}, we know that it suffices to compare the optimal utilities for 
    (1) any $p_{\MO}$ and $q_{\MO} = 0$; (2) $p_{\MO} \in [p_0, \psole)$ and $q_{\MO} =\qthresh(p_{\MO})$;
    (3) $p_{\MO} \in [p_0, \psole)$ and $q_{\MO} = \qthresheps(p_{\MO})$; and (4) $p_{\MO} \in [0, p_0)$ and $q_{\MO} = Q(p_{\MO})$.
    $\MO$'s utility for (1) is 
    $u_{\MO} = (\alpha \psole + k) Q(\psole)$, which is constant in $p_{\MO}$.
The utility in (1) is plotted in green. The utilities for (2) are plotted in red, with a point plotted where the maximal utility is achieved. The same is true for (3), in yellow, and (4), in blue. $p^*_{\MO}$ (circled in red) is then identified by comparing the utility from (1) to the utilities achieved by the maximizing points from (2), (3), and (4) and selecting the price that achieves the highest utility (if (1) has the highest utility, $\MO$ abstains by setting $p^*_{\MO}=\infty$).

\begin{figure}[h!]
    \centering
    \includegraphics[width=.8\textwidth]{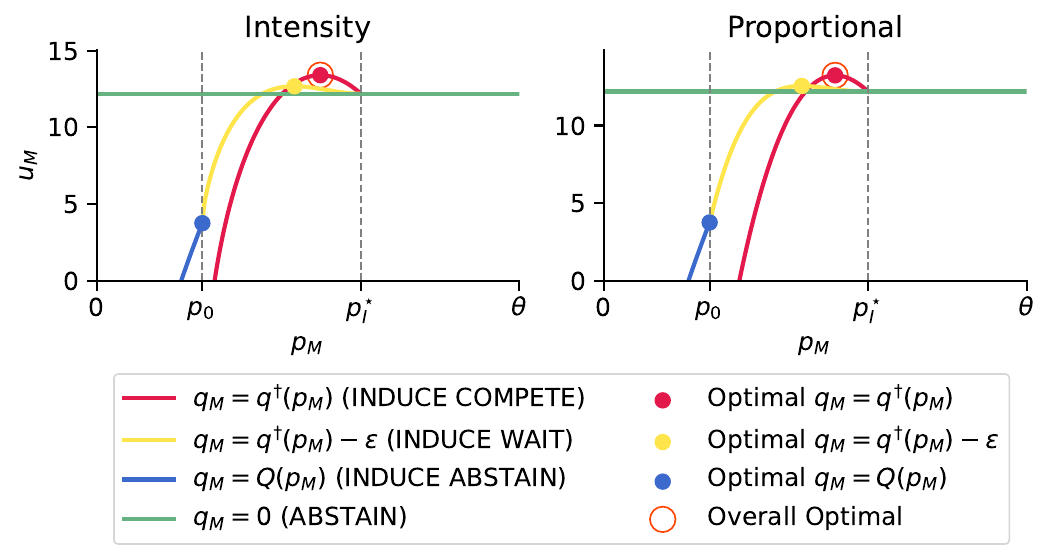}
    \caption{Visualization of $\MO$'s utility for different inventory-setting strategies. Game parameters are $c_{\MO}=4$, $c_{\IS}=2$, $\alpha=0.2$, $k=2$, and $\theta=10$.}
    \label{fig:equilibrium_opt}
\end{figure}

\subsection{Auxiliary results proofs} \label{sec:auxiliary_proofs_APPENDIX}

Here we present the proof of the result that $\IS$'s optimal inventory is determined by their chosen price.

\begin{proof}[Proof of Proposition \ref{prop:optimal_qB}]
    First consider the case where $\IS$ sets $p_{\IS} = \infty$, effectively abstaining from selling. Then the optimal inventory is $q_{\IS} = D_{\IS}(p_{\IS}; q_{\MO}, p_{\MO}) = 0$ because purchasing any positive amount of inventory incurs positive cost but yields no benefit since there is no demand at an infinite price. Now consider the case where $p_{\IS} \in (p_0, \infty)$. It is sub-optimal to set $q_{\IS} < D_{\IS}(p_{\IS}; q_{\MO}, p_{\MO})$ because the net utility $\IS$ gets from selling each unit is constant, so they could have gotten more utility by selling as many units as they are able to at price $p_{\IS}$. It is also sub-optimal to set $q_{\IS} > D_{\IS}(p_{\IS}; q_{\MO}, p_{\MO})$ because $\IS$ incurs a cost for acquiring the excess units but these units yield no positive utility because they are not sold. Finally, consider the case where $p_{\IS} = p_0$. In this case, the seller gets zero net utility per sale, so any $q_{\IS} \in [0, D_{\IS}(p_{\IS};q_{\MO},p_{\MO}))$ is optimal.
\end{proof}

We now prove the result from \Cref{sec:imperfect_substitutability} that it is possible for $\MO$ to de-monopolize $\IS$ even under imperfect substitutability.

\begin{proof}[Proof of \Cref{prop:imperfect_substitutability_general_result}]
    We start by proving a more general result.

\textsc{Fact:} Let $f(x)>0$ and $\delta(x)>0$ for $x>0$ and additionally assume $\delta(x)$ is nondecreasing. Define $g(x) = f(x) - \delta(x)$. Let $x_f^* = \arg \max_x f(x) x$ and $x_g^* = \arg \max_x g(x) x$. Then we must have $$x_f^* > x_g^*.$$

\textsc{Proof of fact}: We will show that the derivative of $g(x)x$ is decreasing at $x_f^*$. 
\begin{align*}
\frac{d}{dx}g(x) x &= g'(x) x + g(x) \\
&= (f’(x) - \delta'(x) )x + f(x) - \delta(x) \\
&= f’(x) x + f(x) - (\delta'(x) x + \delta(x)) \\
&= \frac{d}{dx}f(x) x  - (\delta'(x) x + \delta(x)) \\
\end{align*} 

By the first order optimality condition of $x_f^*$, we get

\begin{align*}
\frac{d}{dx}g(x) x \Big|_{x=x_f^*} = - (\delta’(x_f^*) x_f^* + \delta(x_f^*)) < 0
\end{align*} 
where the inequality comes from noting $x_f^*>0$ and using the fact that $\delta(x)$ is positive and nondecreasing.  

\Cref{prop:imperfect_substitutability_general_result} follows from applying the fact to $f(x) = D(x; 0, \infty)$, which is positive for $x>0$, and $\delta(x) = D(x; 0, \infty) - D(x; \tilde q_{\MO}, \tilde p_{\MO})$, which is positive and nondecreasing by the conditions of the proposition.
\end{proof}


\section{Equilibrium: Practical Implications} \label{sec:applying_equilibrium_APPENDIX}

Having now computed the equilibrium, we can answer questions about $\MO$'s optimal actions and the corresponding equilibrium in different situations. 
For all plots, we use $\theta=10$. All square phase diagrams, here and in the main paper, are constructed by computing the equilibrium for a $200\times 200$ grid of the x-axis and y-axis parameters. 

\paragraph{How should the marketplace operator set $\alpha$?} Suppose that instead of treating the referral fee $\alpha$ as fixed, we allow $\MO$ to choose an $\alpha \in [0,1]$ before the game begins. What value results in the highest utility for the marketplace operator? Figure \ref{fig:uA_vs_alpha} sheds light on this question. 
Even though increasing $\alpha$ increases the revenue $\MO$ gets for each unit sold by $\IS$,
$\MO$ should not in general set $\alpha$ as high as possible. This is because for large values of $\alpha$, the independent seller will only be willing to set a high price and sell few units, or will not be willing to sell at all, so $\MO$ cannot experience the benefits of having $\IS$ in the market. It is worth pointing out that the optimal $\alpha$ is generally \emph{not} the largest $\alpha$ such that $\IS$ competes instead of abstains (represented by the right end of the pink highlighting in Figure \ref{fig:uA_vs_alpha}); the value that maximizes $u_{\MO}$ is somewhat smaller. 

\begin{figure}[h!]
    \begin{subfigure}{.45\textwidth}
      \centering
      \includegraphics[width=\textwidth]{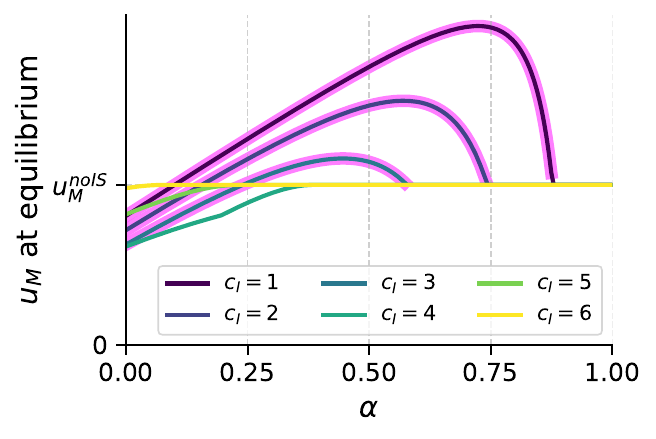}
      \vspace{-15pt}  
      \caption{Intensity}
    \end{subfigure}
    \begin{subfigure}{.45\textwidth}
      \centering
    \includegraphics[width=\textwidth]{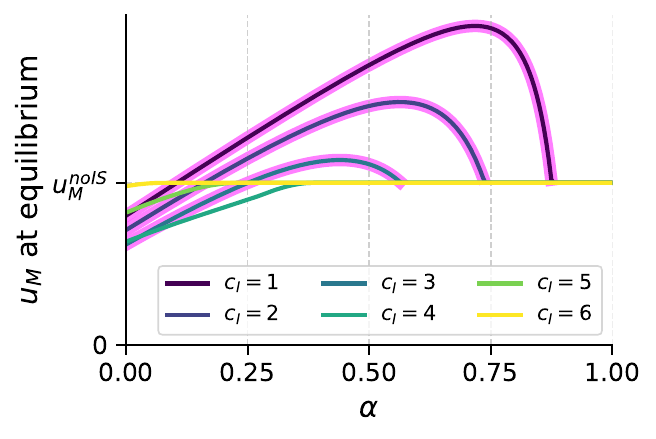}
    \vspace{-15pt}  
      \caption{Proportional}
    \end{subfigure}
    \vspace{-5pt} 
    \caption{$\MO$'s utility at equilibrium as $\alpha$ is varied under different rationing rules. $u_{\MO}^{\text{noIS}}$ marks the utility $\MO$ would receive if $\IS$ were not in the market. Pink highlighting means that $\IS$ competes at the equilibrium induced by those game parameters. 
    The other game parameters are fixed at $c_{\MO} = 5$ and $k=2$.}
    \label{fig:uA_vs_alpha}
\end{figure}

\paragraph{Should $\MO$ behave differently when selling a product with a larger impact on customer experience (large $k$)?} 
Certain products are more important in determining whether a customer is likely to return to the marketplace in the future. 
For example, being able to reliably purchase essential household staples at a marketplace can have a positive effect on customer retention. In our model, such products are modeled as having a large value of $k$. Figure \ref{fig:k_cB_plots_intensity} shows the equilibrium strategies that result from different combinations of $k$ and the independent seller's cost $c_{\IS}$ under intensity rationing. The plots for proportional rationing are similar and are shown in Figure \ref{fig:k_cB_plots_proportional}. Observe that often $q_{\MO}$ increases as $k$ increases. This is true even when $\MO$ is using their inventory to induce $\IS$ to compete rather than selling directly themselves; in such cases, a larger $q_{\MO}$ can induce $\IS$ to set a lower price and sell more inventory. This leads to more satisfied customers, which contributes to marketplace health and benefits $\MO$.

\begin{figure}[h!]
    \begin{subfigure}{.45\textwidth}
      \centering
      \includegraphics[height=1.5in]{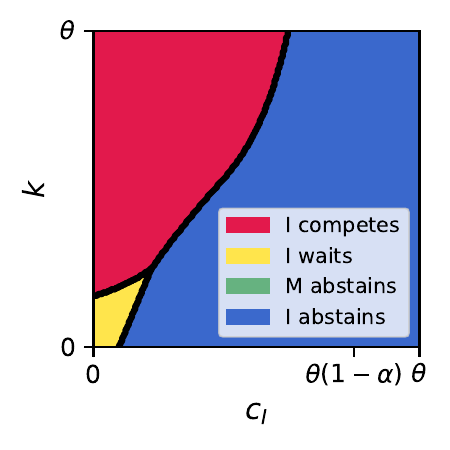}
      \vspace{-10pt}
      \caption{Equilibrium characterization}
    \end{subfigure}
    \begin{subfigure}{.5\textwidth}
      \centering
    \includegraphics[height=1.5in]{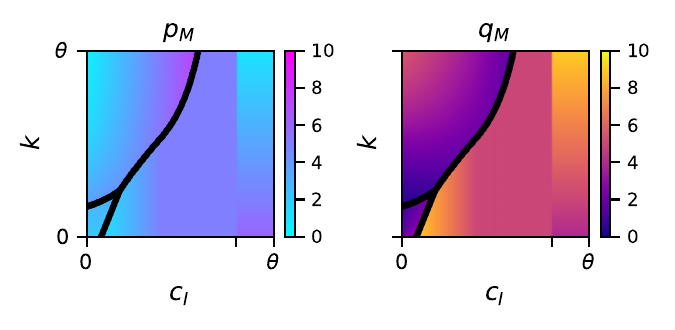}
    \vspace{-10pt}
      \caption{$\MO$'s equilibrium strategies}
    \end{subfigure}
    \caption{Equilibrium characterization and $\MO$'s corresponding strategies for different combinations of $k$ and $c_{\IS}$ under intensity rationing. The other game parameters are fixed at $c_{\MO}=2$ and $\alpha=0.2$.}
    \label{fig:k_cB_plots_intensity}
\end{figure}

\begin{figure}[h!]
    \begin{subfigure}{.45\textwidth}
      \centering
      \includegraphics[height=1.5in]{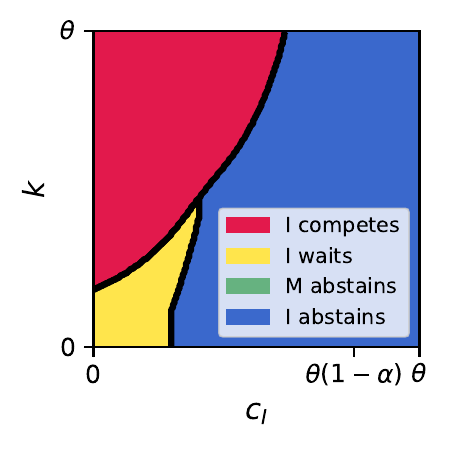}
      \caption{Equilibrium characterization}
    \end{subfigure}
    \begin{subfigure}{.5\textwidth}
      \centering
    \includegraphics[height=1.5in]{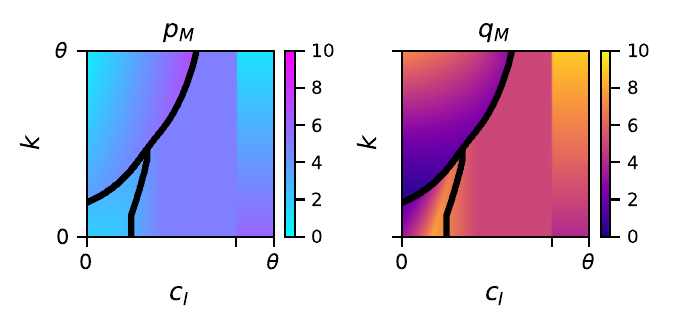}
      \caption{$\MO$'s equilibrium strategies}
    \end{subfigure}
    \caption{Equilibrium characterization and $\MO$'s corresponding strategies for different combinations of $k$ and $c_{\IS}$ under proportional rationing. The other game parameters are fixed at $c_{\MO}=2$ and $\alpha=0.2$.}
    \label{fig:k_cB_plots_proportional}
\end{figure}


\section{Welfare Proofs and Additional Results} \label{sec:welfare_APPENDIX}

This section supplements the content of Section \ref{sec:welfare}. Figure \ref{fig:consumer_surplus} provides an
illustration of the consumer surplus under intensity rationing for different cases.

\begin{figure}[h!]
\centering
\begin{subfigure}{.5\textwidth}
  \centering


\begin{tikzpicture}[scale=0.8]

    \def\thetaval{4.5}
    \def\pB{2}
    \def\qB{2.5}

    \draw[->] (0, 0) -- (0, 5);
    \draw[->] (0, 0) -- (5, 0);

    \node[below] at (2.25,-0.5) {Quantity};
    \node[rotate=90] at (-1,2.25) {Price};

    \draw[thick] (0, \thetaval) -- (\thetaval, 0) node[below] {$\theta$};

    \draw[dashed] (0, \pB) -- (\thetaval, \pB);
    \node[left] at (0, \pB) {$p_{\IS}$};

    \draw[dashed] (\qB, 0) -- (\qB, \thetaval);
    \node[below] at (\qB, 0) {$q_{\IS}$};

    \fill[blue!30, opacity=0.7] (0, \thetaval) -- (0, \pB) -- (\qB, \pB) -- cycle;

    \node[align=center, text=gray] at (\qB+2, \pB+1) {
        \footnotesize $q_{\IS} = D_{\IS}$ assumption\\
        \footnotesize 
        ensures these lines cross
    };
    \draw[->, color=gray, thick] (\qB + 0.5, \pB + 0.5) -- (\qB + .1, \pB + .1);

    \node[left] at (0, 4.5) {$\theta$};
    \node[below] at (4.5, 0) {$\theta$};

\end{tikzpicture}

  \caption{CS for $p_{\IS} \leq p_{\MO}$ and any rationing rule}
\end{subfigure}%
\begin{subfigure}{.5\textwidth}
  \centering


\begin{tikzpicture}[scale=.8]

    \def\thetaval{4.5}
    \def\pB{2}
    \def\qB{2.5}
    \def\pA{1}
    \def\qA{1}

    \draw[->] (0, 0) -- (0, 5);
    \draw[->] (0, 0) -- (5, 0);

    \node[below] at (2.25,-0.5) {Quantity};
    \node[rotate=90] at (-1,2.25) {Price};

    \draw[thick] (0, \thetaval) -- (\thetaval, 0) node[below] {$\theta$};

    \draw[dashed] (0, \pB) -- (\thetaval, \pB);
    \node[left] at (0, \pB) {$p_{\IS}$};
    \draw[dashed] (0, \pA) -- (\thetaval, \pA);
    \node[left] at (0, \pA) {$p_{\MO}$};

    \draw[dashed] (\qB, 0) -- (\qB, \thetaval);
    \node[below] at (\qB, 0) {$q_{\IS}$};
    \draw[dashed] (\qA, 0) -- (\qA, \thetaval);
    \node[below] at (\qA, 0) {$q_{\MO}$};

    \fill[green!30, opacity=0.7] (0, \thetaval) -- (0, \pA) -- (\qA, \pA) -- (\qA, \thetaval - \qA) -- cycle;

    \fill[blue!30, opacity=0.7] (\qA, \thetaval - \qA) -- (\qA, \pB) -- (\qB, \pB) -- cycle;

    \node[align=center, text=gray] at (\qB+2, \pB+1) {
        \footnotesize $q_{\IS} = D_{\IS}$ assumption\\
        \footnotesize 
        ensures these lines cross
    };
    \draw[->, color=gray, thick] (\qB + 0.5, \pB + 0.5) -- (\qB + .1, \pB + .1);

    \node[left] at (0, 4.5) {$\theta$};
    \node[below] at (4.5, 0) {$\theta$};

\end{tikzpicture}

  \caption{CS for $p_{\IS} \geq q_{\MO}$ and intensity rationing}
\end{subfigure}
\caption{A visualization of consumer surplus under the assumption that $q_{\IS} = D_{\IS}$, which is true whenever $\IS$ best responds. Blue denotes consumer surplus due to $\IS$'s sales. Green denotes consumer surplus due to $\MO$'s sales.}
\label{fig:consumer_surplus}
\end{figure}

\begin{proof}[Proof of \cref{lemma:surplus_transfer}]
    Fix any $p_{\MO}$ and $q_{\MO}$. For notational brevity, we will write $u_{\IS}$ for $u_{\IS}\big(p_{\MO}, q_{\MO}, \pBR(p_{\MO}, q_{\MO}), \qBR(p_{\MO}, q_{\MO})\big)$ and $\CS$ for $\CS\big(p_{\MO}, q_{\MO}, \pBR(p_{\MO}, q_{\MO}), \qBR(p_{\MO}, q_{\MO})\big)$. To show $- \Delta \PS \leq \Delta \CS$, it is equivalent to show 
    \begin{align} \label{eq:uIS_CS_goal_inequality}
        \uISsole + \CSsole \leq u_{\IS} + \CS.
    \end{align}
    Note that whenever $\IS$ is the sole seller or is competing, their utility for setting price $p_{\IS}$ and quantity $Q(p_{\IS})$ can be expressed as
    \begin{align} \label{eq:uIS_integral}
        \int_{0}^{Q(p_{\IS})} (1-\alpha) p_{\IS} - c_{\IS}~dt,
    \end{align}
    and if they wait, their utility for setting price $p_{\IS}$ and quantity $Q(p_{\IS}) - q_{\MO}$ can be expressed as
    \begin{align} \label{eq:uIS_integral_wait}
        \int_{q_{\MO}}^{Q(p_{\IS})} (1-\alpha) p_{\IS} - c_{\IS}~dt.
    \end{align}
    Evaluating \eqref{eq:CS_one_seller} and \eqref{eq:uIS_integral} at $\psole$ and combining allows us to express the left-hand side of \eqref{eq:uIS_CS_goal_inequality} as
    \begin{align}
        \uISsole + \CSsole = \int_{0}^{Q(\psole)} Q(t) - \alpha \psole - c_{\IS}~dt.
    \end{align}
    We now have to obtain a comparable expression for the right-hand side. As described in \cref{lemma:best_response_intensity}, $\IS$'s best response looks different depending on the values of $p_{\MO}$ and $q_{\MO}$.
    We split our analysis into cases depending on $\IS$'s best response.

    \emph{Case 1: $\IS$'s best response is to set their monopolist price $\psole$.} \eqref{eq:uIS_CS_goal_inequality} is trivially satisfied with equality in this case.

    \emph{Case 2. $\IS$'s best response is to compete by setting $p_{\IS} \leq \psole$.} By \eqref{eq:uIS_integral}, we have
    $$u_{\IS} + \CS = \int_0 ^{Q(p_{\IS})} Q(t) - \alpha p_{\IS} - c_{\IS}~dt, $$
    which is greater than or equal to $\uISsole + \CSsole$ since $p_{\IS} < \psole$ whenever $\IS$'s best response is to compete.

    \emph{Case 3. $\IS$'s best response is to wait.} By \eqref{eq:CS_two_sellers} and \eqref{eq:uIS_integral_wait}, we have
    \begin{align*}
        u_{\IS} + \CS = \int_0^{q_{\MO}} Q(t) - p_{\MO} ~dt + \int_{q_{\MO}} ^{Q(p_{\IS})} Q(t) - \alpha p_{\IS} - c_{\IS}~dt.
    \end{align*}
    This is greater than or equal to $\uISsole + \CSsole$ because $p_{\MO} < \psole$ whenever $\IS$'s best response is to compete.

    \emph{Case 4. $\IS$'s best response is to abstain.} In this case, $u_{\IS} + \CS = \CS = \int_0^{q_{\MO}} Q(t) - p_{\MO} ~dt \geq \uISsole + \CSsole$. The last inequality follows from the fact that $\IS$'s best response is abstain only if $p_{\MO} < p_0$ and $q_{\MO} > \theta - p_0 > Q(\psole)$.
\end{proof}

\begin{proof} [Proof of \cref{prop:CS_increases}]
We break this into three cases. 

\emph{Case 1:} $p_{\MO} \geq \psole$. This is effectively the same as if the marketplace operator were not participating because the independent seller will simply set $p_{\IS} = \psole$ as they would have had the marketplace operator not been a seller (Proposition \ref{prop:best_response_large_pA}). Thus, the consumer surplus with the marketplace operator is equal to without. 

\emph{Case 2: $p_{\MO} < \psole$ and the independent seller abstains.} By Proposition \ref{prop:best_response_small_pA}, we know that under intensity rationing, the independent seller abstains only if $q_{\MO} \geq \theta - p_0 (\geq Q(\psole))$. Thus, we can be in this case only if $\MO$ has offered more inventory compared to when $\IS$ is the sole seller and at a lower price.  
In other words, $p_{\MO} < \psole$ and $q_{\MO} \geq Q(\psole)$, which implies the consumer surplus must be larger than $\CSsole$.

\emph{Case 3: $p_{\MO} < \psole$ and the independent seller competes or waits.} By Propositions \ref{prop:best_response_intermediate_pA} and \ref{prop:best_response_small_pA}, we know that under intensity rationing, the independent seller will set  $p_{\IS} \leq p_{\MO} (< \psole)$ whenever they are not abstaining. Furthermore, by Proposition \ref{prop:optimal_qB}, we know that the independent seller will fully meet demand, so the effect is that there is more total inventory at a lower price. This necessarily increases the consumer surplus. We visualize this in Figure \ref{fig:CS_increases_proof}, which shows that regardless of whether $\IS$ waits or competes, the CS when $\MO$ is a seller (represented by the striped orange and magenta regions) is always larger than $\CSsole$ (represented by the solid blue triangle). 
\end{proof}

\begin{figure}[h]
\centering
\begin{subfigure}{.45\textwidth}
  \centering
  \resizebox{1.8in}{!} {



\begin{tikzpicture}[scale=1]

    \def\thetaval{4.5}
    \def\pB{1.5}
    \def\qB{3}
    \def\pA{1}
    \def\qA{1}
    \def\pone{2.3}

    \draw[->] (0, 0) -- (0, 5);
    \draw[->] (0, 0) -- (5, 0);

    \node[below] at (2.25,-0.5) {Quantity};
    \node[rotate=90] at (-1,2.25) {Price};

    \fill[color=blue,  opacity=0.7] (0, \thetaval) -- (\thetaval - \pone, \pone) -- (0, \pone) -- (0, \thetaval);

    \draw[thick] (0, \thetaval) -- (\thetaval, 0) node[below] {$\theta$};

    \draw[dashed] (0, \pB) -- (\thetaval, \pB);
    \node[left] at (0, \pB) {$p_{\IS}$};
    \draw[dashed] (0, \pA) -- (\thetaval, \pA);
    \node[left] at (0, \pA) {$p_{\MO}$};
    \draw[dashed] (0, \pone) -- (\thetaval, \pone);
    \node[left] at (0, \pone) {$\psole$};

    \draw[dashed] (\qB, 0) -- (\qB, \thetaval);
    \node[below] at (\qB, 0) {$q_{\IS}$};
    \draw[dashed] (\qA, 0) -- (\qA, \thetaval);
    \node[below] at (\qA, 0) {$q_{\MO}$};
    \draw[dashed] (\thetaval - \pone, 0) -- (\thetaval - \pone, \thetaval);
    \node[below] at (\thetaval - \pone, 0) {$Q(\psole)$};

    \filldraw[
    draw=orange,              
    pattern=north east lines, 
    pattern color=orange,     
    ] (0, \thetaval) -- (0, \pA) -- (\qA, \pA) -- (\qA, \thetaval - \qA) -- cycle;

     \filldraw[
    draw=magenta,              
    pattern=north east lines, 
    pattern color=magenta,     
    ](\qA, \thetaval - \qA) -- (\qA, \pB) -- (\qB, \pB) -- cycle;


    \node[left] at (0, 4.5) {$\theta$};
    \node[below] at (4.5, 0) {$\theta$};

\end{tikzpicture}

  }
  \caption{$\IS$ waits}
  \label{fig:CS_increases_proof_B_waits}
\end{subfigure}
\begin{subfigure}{.45\textwidth}
  \centering
  \resizebox{1.8in}{!}{



\begin{tikzpicture}[scale=1]

    \def\thetaval{4.5}
    \def\pB{1.5}
    \def\qB{3.5}
    \def\pA{1}
    \def\qA{1}
    \def\pone{2.3}

    \draw[->] (0, 0) -- (0, 5);
    \draw[->] (0, 0) -- (5, 0);

    \node[below] at (2.25,-0.5) {Quantity};
    \node[rotate=90] at (-1,2.25) {Price};

    \fill[color=blue,  opacity=0.7] (0, \thetaval) -- (\thetaval - \pone, \pone) -- (0, \pone) -- (0, \thetaval);

    \draw[thick] (0, \thetaval) -- (\thetaval, 0) node[below] {$\theta$};

    \draw[dashed] (0, \pA) -- (\thetaval, \pA);
    \node[left] at (0, \pA) {$p_{\IS}=p_{\MO}$};
    \draw[dashed] (0, \pone) -- (\thetaval, \pone);
    \node[left] at (0, \pone) {$\psole$};

    \draw[dashed] (\qB, 0) -- (\qB, \thetaval);
    \node[below] at (\qB, 0) {$q_{\IS}$};
    \draw[dashed] (\thetaval - \pone, 0) -- (\thetaval - \pone, \thetaval);
    \node[below] at (\thetaval - \pone, 0) {$Q(\psole)$};


     \filldraw[
    draw=magenta,              
    pattern=north east lines, 
    pattern color=magenta,     
    ](0, \thetaval) -- (0, \pA) -- (\thetaval - \qA, \pA) -- cycle;


    \node[left] at (0, 4.5) {$\theta$};
    \node[below] at (4.5, 0) {$\theta$};

\end{tikzpicture}

}
  \caption{$\IS$ competes}
  \label{fig:CS_increases_proof_B_competes}
\end{subfigure}
\caption{\textcolor{blue}{Blue} denotes the consumer surplus when the independent seller is the sole seller. The striped regions denote consumer surplus when the marketplace operator is also a seller, where the \textcolor{orange}{orange} and \textcolor{magenta}{magenta} denote consumer surplus due to the operator's and independent seller's sales, respectively.}

\label{fig:CS_increases_proof}
\end{figure}

\paragraph{Welfare.} We now turn our attention to total welfare, which can be computed as
\begin{align}
    W(p_{\MO}, q_{\MO}, p_{\IS}, q_{\IS}) = \mathrm{CS}(p_{\MO}, q_{\MO}, p_{\IS}, q_{\IS}) + u_{\MO}(p_{\MO}, q_{\MO}, p_{\IS}, q_{\IS}) + u_{\IS}(p_{\MO}, q_{\MO}, p_{\IS}, q_{\IS}).
\end{align}

Figure \ref{fig:welfare_all_combinations} provides a comprehensive visualization of the difference in welfare that results from $\MO$'s participation in the game.
For all of these plots, the following observations hold: (1) The total welfare when both $\MO$ and $\IS$ can choose to sell is always at least as high as the welfare when $\IS$ is the sole seller. (2) The largest gain in total welfare occurs when $c_{\IS}$ is large and $c_{\MO}$ is small. This is not surprising because this is the region in which the equilibrium result is that $\MO$ is the only seller. 

\begin{figure}[h!]
    \centering
    \includegraphics[width=.9\linewidth]{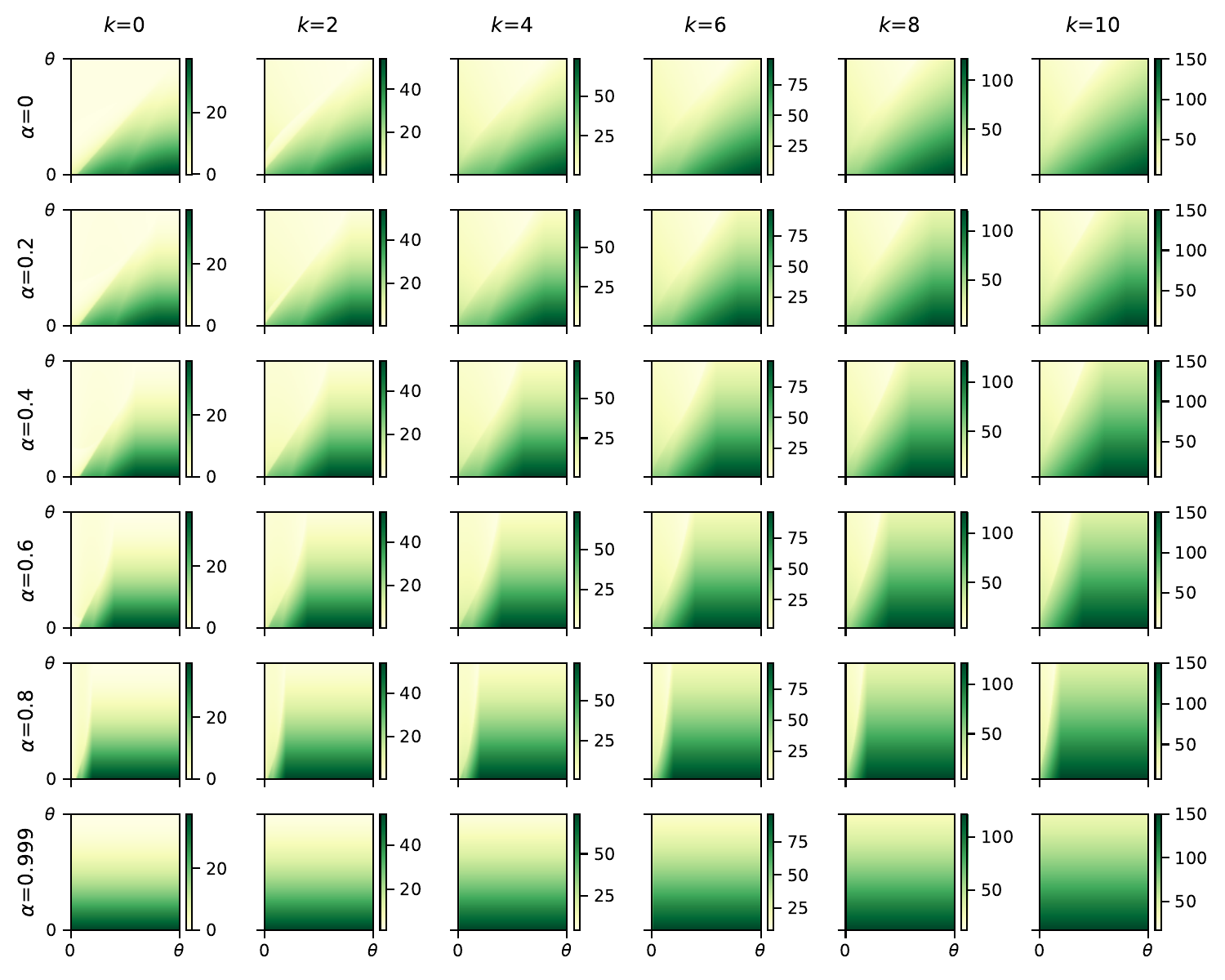}
    \caption{Difference in the equilibrium welfare and the welfare when $\IS$ is the only seller for different combinations of game parameters ($W(p_{\MO}, q_{\MO}, p_{\IS}, q_{\IS}) - W(\infty, 0, p_{\IS}, q_{\IS})$). Each individual plot varies over $c_{\IS}$ on the x-axis and $c_{\MO}$ on the y-axis. Each row of plots corresponds to a value of $\alpha \in \{0, 0.2, 0.4, 0.6, 0.8, 0.999\}$ and each column corresponds to a value of $k\in\{0,2,4,6,8,10\}$. We use $\alpha=0.999$ rather than $\alpha=1$ in the bottom row because the ``$\IS$ only" game is not well-defined when $\alpha=1$ because this would mean $\IS$ must pay 100\% of their revenue to $\MO$.
    The minimum and maximum values of each colorbar correspond to the minimum and maximum values in the data being displayed.}
    \label{fig:welfare_all_combinations}
\end{figure}


\section{Discussion of Rationing Rules} \label{sec:rationing_rule_APPENDIX}

In this section we provide an in-depth discussion of rationing rules. In Section \ref{sec:rationing_rule_background_APPENDIX} we provide relevant background on demand functions and explain why we choose to not study inverse intensity rationing. In Section \ref{sec:rationing_rules_examples_APPENDIX}, we provide examples of intensity, proportional, and inverse intensity rationing. In Section \ref{sec:prop_rationing_derivation_APPENDIX}, we describe the probabilistic modeling assumptions underlying proportional rationing and show how the residual demand function is derived starting from these assumptions. In Section \ref{sec:realism_of_proportional}, we show that, although the assumptions of proportional rationing are not entirely natural, the implied residual demand function is very similar to the residual demand function that we arrive at under more natural assumptions.

\subsection{Background}\label{sec:rationing_rule_background_APPENDIX}

For additional discussions of rationing rules, see Section 14.2 of \citet{rasmusen1989games}, Section 5.3.1 of \citet{tirole1988theory}, and Section 2 of \citet{davidson1986long}. The last describes a perspective on demand functions and rationing rules that views demand functions as arising from individual consumers each with their own identical demand curve. We will instead view demand functions as arising a population of customers with heterogeneous valuations who each wish to purchase one unit of the good, as presented in Section 4.1.1 of \citet{vohra2012principles}.

\emph{Demand functions.} Before describing rationing rules, we will review some background on demand functions. A demand function $Q(p)$ tells you the \emph{expected} number of people who will buy a good at price $p$. Specifically, assume that there are $n$ total customers, and each customer $i$ has a valuation $v_i$ (a.k.a.\ the maximum price they are willing to pay, a.k.a.\ their reservation price) that is an independent random variable sampled iid from some common distribution. Then the demand function is
\begin{align*}
    Q(p) &= \sum_{i=1}^n \E[\I{v_i \geq p}] \\
    &= \sum_{i=1}^n \P(v_i \geq p) \\
    &= n \P(v_1 \geq p) \quad \text{since $v_i$'s are iid}.
\end{align*}

If we assume that the distribution of the valuations is \emph{uniform}, this induces a linear demand function: if $v_i \sim \text{Unif}([0,\theta])$, then $\P(v_1 \geq p) = (\theta - p)/\theta$ for $p \in [0,\theta]$. This allows us to write the demand function as
\begin{align*}
    Q(p, \theta) = \begin{cases}
         n - \frac{n}{\theta}p & \text{for $0 \leq p \leq \theta$} \\
        0  & \text{for $p>\theta$}.
    \end{cases}
\end{align*}
If $n = \theta$, this further simplifies to $Q(p) = \theta - p$, which is the demand function we consider, as described in Assumption \ref{assumption:linear_demand}.

\emph{A remark on inverse intensity rationing.} 
One type of rationing that has appeared in the literature but we do \emph{not} consider is inverse intensity rationing, also known as inefficient rationing \citep{rasmusen1989games}. This rationing scheme assumes that the \emph{lowest valuation customers buy first}, yielding the residual demand function 
$$R(p_i) = \min\big(Q(p_i), \theta - q_j\big).$$
This is visualized in Figure \ref{fig:R_inverse_intensity}.
We choose not to study this type of rationing rule due to its logical weaknesses. When we think about the inverse intensity rationing from the perspective of individual buyers, the underlying assumption seems to imply irrational behavior. Specifically, it requires that we assume that the buyers with the lowest valuation buy the good first, \emph{even if their valuation does not exceed the price}. This leads to weird behavior when we attempt to use this rationing rule in our analysis, because it implies the creation of demand that is not captured by the original demand curve.
A more realistic alternative to inverse intensity rationing would be to assume that the customers with the lowest valuation \emph{above} the price are the first to buy. This would lead to an $R$ that has a vertical drop from price $p_i+q_j$ down to $p_i$ and then has slope $-1$ from price $p$ down to zero. In other words, it looks like $Q(p)$ until price $p_i+q_j$ and then matches intensity rationing from price $p_i$ down.

\begin{figure}
    \centering
    \begin{tikzpicture}
    \draw[->] (0,0) -- (4.5,0);
    \draw[->] (0,0) -- (0,4.5);

    \node[below] at (2.25,-0.5) {Quantity};
    \node[rotate=90] at (-0.8,2.25) {Price};

    \def\thetaval{4}
    \def\qA{1}

    \draw[thick, domain=0:\thetaval] plot (\x,{\thetaval - \x});

    \draw[blue, very thick, dashed] (0, \thetaval) -- (\thetaval-\qA, \qA) -- (\thetaval-\qA, 0);
    

    \node at (\thetaval,-0.35) {$\theta$};
    \node at (-0.3,\thetaval) {$\theta$};
    \node at (2,2.6) {$Q(p)$};
    \node[blue] at (1.8,1.6) {$R(p)$};
    \node[blue] at (\thetaval - \qA, -0.3) {$\theta - q_{\MO}$};
    \node[white] at (3,-.35) {};
\end{tikzpicture}
    \caption{Visualization of the residual demand function $R(p_{\IS})$ for $\theta=4$, $q_{\MO}=1$, and $p_{\MO}=2$ under inverse intensity rationing.}
    \label{fig:R_inverse_intensity}
\end{figure}

\subsection{Examples of each rationing rule} \label{sec:rationing_rules_examples_APPENDIX}

Recall that rationing rules are essentially an assumption about the order in which customers arrive. The customers who arrive earliest have first pick between the lower-priced seller and higher-priced seller and since we assume that the customers are utility maximizing, they will purchase from the lower-priced seller as long as they have inventory.
Intensity rationing assumes that customers arrive in order of descending willingness to pay. Proportional rationing (roughly) assumes that the arrival order is independent of valuation. Inverse intensity rationing assumes that customers arrive in order of ascending willingness to pay. 

\emph{Example of intensity rationing.} 
Customers who want a product the most will (if the product is sold online) refresh the page most recently or (if the product is sold in person), be willing to wait in a long line on the product launch date. Thus the customers with the highest willingness to pay get the earliest access to the product and are able to buy the lower-priced version before it goes out of stock. For example, people wait in line to buy iPhones on their release dates \citep{yahooIphone13}.

\emph{Example of proportional rationing.} Before presenting our example, we remark that one peculiarity about proportional rationing is that the probability that a customer is routed to the lower-priced seller is increasing in the lower-priced seller's inventory. Once a customer is routed to seller, they must choose whether to buy from that seller and have no ability to switch to the other seller. We now present a stylized example that satisfies this: suppose there are two hot dog chains, Cheap Dogs and Fancy Dogs. Each chain has many hot dog carts scattered across the city and each chain's inventory is proportional to their number of carts. People walk around the city  and when they get hungry, they buy a hotdog from the first stand they encounter. If the first stand they go to does not offer a hotdog a price they are happy with, they become so dejected that they no longer want to eat a hotdog, so there is no transfer of demand from a higher-priced seller to a lower-priced seller.

\emph{Example of inverse-intensity rationing.} Suppose a new product is being released and consumers initially do not know whether it is worth buying, so the first consumers to buy the good are not willing to pay much for it. However, once the early adopters ascertain that the good is high quality, they spread the word to other consumers who are then willing to pay more for the product. This is the reason that new companies and products frequently offer coupons to attract customers.

\subsection{Derivation of residual demand function under proportional rationing} \label{sec:prop_rationing_derivation_APPENDIX}

In this section, we describe the probabilistic modeling assumptions that underlie proportional rationing. Under proportional rationing,
each customer is randomly assigned to the higher price with some probability $\rho$ and the lower price with probability $1-\rho$. There is a total population of $n=\theta$ customers. $\rho$ is such that in expectation the demand that the lower-priced seller faces at their price $p_j$ is exactly equal to their supply $q_j$, which is
    $$\rho = \frac{Q(p_j) - q_j}{Q(p_j)}.$$
    The full probabilistic model is: There are $\theta$ customers each with valuation $v_i \sim \mathrm{Unif}([0,\theta])$. Independent of their valuation, the customer is sent to the higher-priced seller with probability $\rho$; otherwise they are sent to the lower-priced seller. Customers either buy at the seller they are sent to or do not buy at all. Given this model, the expected demand that the higher-priced seller faces when setting price $p$ and the lower-priced seller sets price $p_j$ and quantity $q_j$ is
    \begin{align*}
    R(p; q_j, p_j) &= \E\left[\sum_{i=1}^{n} \I{v_i \geq p} \I{i\text{ assigned to high price seller}} \right] \\
    &= \sum_{i=1}^{n} \E\left[\I{v_i \geq p} \I{i\text{ assigned to high price seller}} \right] \qquad \text{linearity of expectation}\\
    &= \sum_{i=1}^{n} \E[\I{v_i \geq p}] \E[\I{i\text{ assigned to high price seller}}] \qquad \text{independence}\\
    &=  \sum_{i=1}^{n} \P(v_i \geq p)\P(i\text{ assigned to high price seller}) \\
    &= n \frac{\theta-p}{\theta} \rho \\
    &= (\theta-p)\rho = (\theta-p)\frac{Q(p_j) - q_j}{Q(p_j)} = Q(p) \frac{Q(p_j) - q_j}{Q(p_j)} \qquad \text{since $n=\theta$}.
    \end{align*}

\subsection{Realism of proportional rationing} \label{sec:realism_of_proportional}

As previously described, the assumptions underpinning proportional rationing are a bit peculiar. We adopt this rationing rule in the main paper because it leads to an analytical residual demand function that is easy to work with and is also the convention in existing literature.
In this section, we analyze an alternative to proportional rationing based on more realistic assumptions and find it generally produces residual demand functions that are very similar to those under proportional rationing. 

\emph{More realistic assumptions.} There are $n=\theta$ total customers. Assume that customers arrive one at a time and their valuation is sampled $v_i \sim \text{Unif}([0,\theta])$. A customer purchases the product if the lowest price available does not exceed their valuation. Otherwise, they leave immediately (whether they leave or wait for a lower price does not matter in practice, because the price only ever goes up after one seller runs out). 

We will derive the residual demand function for the higher-priced seller. Let $\text{NegBin}(r,p)$ denote the probability distribution corresponding to the number of trials needed to get $r$ successes from flipping a coin with bias $p$, which has probability mass function $f(k) = {k-1 \choose r-1}p^r (1-p)^{k-r}$. 
Let $N \sim \text{NegBin}(q_j,\frac{\theta - p_j}{\theta})$ be the number of customers that would have to arrive in order for all of the lower-priced seller's inventory to be exhausted. 
Let $M = \max(\theta - N, 0)$ be the number of buyers who have not yet visited the website after the lower-priced seller sells out and let $V_1, ..., V_M$ denote the valuations of the customers that arrive after the lower-priced seller has sold out.
Then the residual demand is
\begin{align*}
    R(p) &= \E_{M}\left[\E_{V_1,..., V_M}\left[\sum_{i=1}^M\I{V_i \geq p} \Big| M \right] \right] \quad  \text{by tower property}\\
    &= \E_{M}\left[{M \frac{\theta -p}{\theta}}\right] \quad \text{by linearity of expectation}\\ 
    &= \frac{\theta -p}{\theta} \E[M] \\
    &= \frac{\theta -p}{\theta} \E[\max(\theta - N, 0)] \\
    &=  \frac{\theta - p}{\theta} \sum_{k=q_j}^{\theta} (\theta - k) {k - 1 \choose q_j - 1} \left(\frac{p_j}{\theta}\right)^{k-q_j} \left(\frac{\theta - p_j}{\theta}\right)^{q_j}.
\end{align*}

Figure \ref{fig:proportional_vs_random_arrivals} illustrates that when the lower-priced seller's quantity $q_j$ is small relative to demand at price their $p_j$, the residual demand under our more realistic assumptions is visually identical to under proportional rationing. The difference between the two curves only becomes pronounced when $q_j$ approaches $Q(p_j)$. 

\begin{figure}
    \centering
    \includegraphics[width=\textwidth]{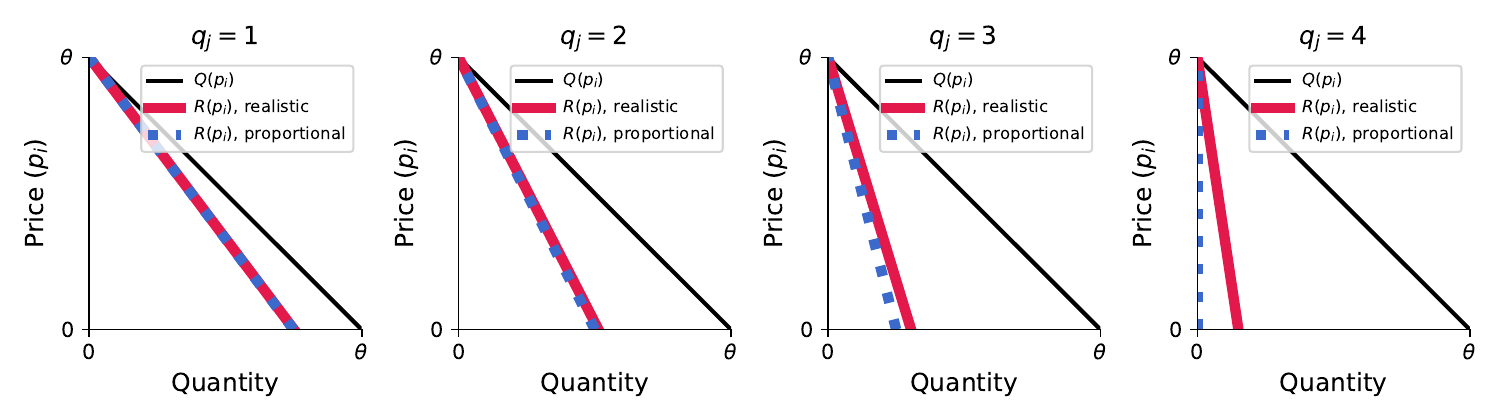}
    \caption{Comparison of residual demand function under proportional rationing and under the more realistic assumptions described in Appendix \ref{sec:realism_of_proportional}. $p_j$ is set to 6, so the corresponding original demand is $Q(p_j)=\theta - 6 = 4$.}
    \label{fig:proportional_vs_random_arrivals}
\end{figure}

\end{document}